\documentclass{acmsmall}
\usepackage{hyperref}
\usepackage{url}
\usepackage[numbers]{natbib}
\usepackage{amssymb}
\usepackage{helvet}
\usepackage{courier}
\usepackage{graphicx}
\usepackage{ifthen}
\usepackage{amsmath,amssymb}
\usepackage{eufrak}
\usepackage {float}
\usepackage[amsmath,thmmarks]{ntheorem}
\usepackage{mathrsfs}
\usepackage{verbatim}
\usepackage{array,color}
\usepackage[table]{xcolor}
\usepackage[ruled,linesnumbered,boxed]{algorithm2e}
\usepackage{multirow}
\usepackage{url}

\newcommand{\bv}{\begin{array}}

\newcommand{\kendall}{\text{Kendall}}

\newtheorem{thm}{Theorem}
\newtheorem{dfn}{Definition}

\newtheorem{lem}{Lemma}
\newtheorem{ex}{Example}

\newtheorem{prop}{Proposition}

\newtheorem{claim}{Claim}
\renewenvironment{proof}{\noindent{\bf Proof:}{}}{\hfill $\blacksquare$ }

\newcommand{\rank}{\text{Rank}}

\newcommand{\Omit}[1]{}

\newcommand{\mn}{\mathcal N}

\newcommand{\md}{\mathfrak D}
\newcommand{\mo}{\mathcal O}
\newcommand{\mk}{\mathcal K}

\newcommand\boxit[1]{{\begin{tabular}{|@{\hspace{.5mm}}c@{\hspace{.5mm}}|}\hline$#1$\\ \hline\end{tabular}} }

\makeatletter
\def\runningfoot{\def\@runningfoot{}}

\def\firstfoot{\def\@firstfoot{}}
\makeatother

\begin{document}

\title{Allocating Indivisible Items in Categorized Domains}
\author{Erika Mackin\affil{mackie2@rpi.edu, 
Computer Science Department, Rensselaer Polytechnic Institute,
110 8th Street, Troy, NY 12180, USA
} Lirong Xia\affil{xial@cs.rpi.edu, 
Computer Science Department, Rensselaer Polytechnic Institute,
110 8th Street, Troy, NY 12180, USA
}}

\begin{abstract} 
We formulate a general class of allocation problems called {\em categorized domain allocation problems (CDAPs)}, where indivisible items from multiple categories are allocated to agents without monetary transfer and each agent gets at least one item per category. 

We focus on {\em basic CDAPs}, where the number of items in each category is equal to the number of agents. We characterize serial dictatorships for basic CDAPs by a minimal set of three axiomatic properties: {strategy-proofness}, {non-bossiness}, and {category-wise neutrality}. Then, we propose a natural extension of serial dictatorships called {\em categorial sequential allocation mechanisms (CSAMs)}, which allocate the items in multiple rounds: in each round, the active agent chooses an item from a designated category. We fully characterize the worst-case {\em rank efficiency} of CSAMs for optimistic and pessimistic agents, and provide a bound for strategic agents. We also conduct experiments to compare expected rank efficiency of various CSAMs w.r.t.~random generated data.
\end{abstract}
%

\maketitle

\section{Introduction}
Suppose we are organizing a seminar and must allocate $10$ discussion topics and $10$ dates to $10$ students. Students have heterogeneous and combinatorial preferences over (topic, date) bundles: a student's preferences over the topics may depend on the date and vice versa, because she may prefer an early date if she gets an easy topic and may prefer a late date if she gets a hard topic.

This example illustrates a common setting for allocating multiple indivisible items, which we formulate as  a {\em categorized domain}. A categorized domain contains multiple indivisible items, each of which belongs to one of the $p\geq 1$ categories. In {\em categorized domain allocation problems (CDAPs)}, we want to design a mechanism to allocate the items to agents without monetary transfer, such that each agent gets at least one item per category. In the above example, there are two categories: topics and dates, and each agent (student) must get a topic and a date. 

Many other allocation problems are CDAPs. For example, in cloud computing, agents have heterogeneous preferences over multiple types of items including CPU, memory, and storage\footnote{Suppose each type contains discrete units of resources that are essentially indivisible for operational convenience.}~\cite{Ghodsi11:Dominant,Ghodsi12:Multi,Bhattacharya13:Hierarchical}; patients must be allocated multiple types of resources including surgeons, nurses, rooms, and equipment~\cite{Huh13:Multiresource}; college students need to choose courses from multiple categories per semester, e.g.~computer science courses, math courses, social science courses, etc. 


The design and analysis of allocation mechanisms for non-categorized domains have been an active research area at the interface of computer science and economics.
In computer science, allocation problems have been studied as {\em multi-agent resource allocation}~\cite{Chevaleyre06:Issues}. 
In economics, allocation problems have been studied as {\em one-sided matching}, also known as  {\em assignment problems}~\cite{Sonmez11:Matching}. 
Previous research faces three main barriers.

$\bullet$  {\em Preference bottleneck:} When the number of items is not too small, it is impractical for the agents to express their preferences over all (exponential) bundles of items.

$\bullet$  {\em Computational bottleneck:} Even if the agents can express their preferences compactly using some preference language, computing an ``optimal'' allocation is often a hard combinatorial optimization problem. 

$\bullet$  {\em Threats of agents' strategic behavior:} An agent may have incentive to report untruthfully to obtain a more preferred bundle. This may lead to a socially inefficient allocation.


\vspace{2mm}
{\noindent\bf Our Contributions.} We initiate the study of mechanism design under the novel framework of CDAPs towards breaking the three aforementioned barriers. CDAPs naturally generalize classical non-categorized allocation problems, which are CDAPs with one category. CDAPs are our main conceptual contribution.

As a first step, we focus on {\em basic categorized domain allocation problems (basic CDAPs)}, where the number of items in each category is exactly the same as the number of agents, so that each agent gets exactly one item from each category. See e.g.~the seminar-organization example. 

Our technical contributions are two-fold. First, we characterize {\em serial dictatorships} for any basic CDAPs with at least two categories by a minimal set of three axiomatic properties: {\em strategy-proofness}, {\em non-bossiness}, and {\em category-wise neutrality}. This helps us understand the possibility of designing strategy-proof mechanisms to overcome the third barrier, i.e.~threats of agents' strategic behavior. 

Second, to overcome the preference bottleneck and the computational bottleneck, and to go beyond serial dictatorships, we propose  
{\em categorial sequential allocation mechanisms (CSAMs)}, which are a large class of indirect mechanisms that naturally extend serial dictatorships~\cite{Svensson99:Strategy-proof}, {sequential allocation protocols}~\cite{Bouveret11:General}, and the draft mechanism~\cite{Budish12:Multi}. For $n$ agents and $p$ categories, a CSAM is defined by an ordering over all (agent, category) pairs: in each round, the active agent picks an item that has not been chosen yet from the designated category. CSAMs have low communication complexity and low computational complexity.

We completely characterize the worst-case {\em rank efficiency} of CSAMs, measured by agents' {\em ranks} of the bundles they receive, for any combination of two types of {\em myopic} agents:  {\em optimistic} agents, who always choose the item in their top-ranked bundle that is still available, and {\em pessimistic} agents, who always choose the item that gives them best worst-case guarantee.
This characterization naturally leads to useful corollaries on worst-case efficiency of various CSAMs. For example, we show that while serial dictatorships with all-optimistic agents have the best worst-case utilitarian rank, they have the worst worst-case egalitarian rank (Proposition~\ref{prop:optsd}). On the other hand, balanced CSAMs with all-pessimistic agents have good worst-case egalitarian rank (Proposition~\ref{prop:bcsm}). For strategic agents, we prove an upper bound on the worst-case rank efficiency for all CSAMs, and a matching lower bound for two agents.

We also use computer simulation  to further compare expected utilitarian rank and expected egalitarian rank of some natural CSAMs where agents' preferences are generated from the {\em Mallows model}~\cite{Mallows57:Non-null}. 
We observe that serial dictatorships with all-optimistic agents have good expected utilitarian rank and bad expected egalitarian rank. On the other hand, the balanced CSAMs with all-pessimistic agents have good expected egalitarian rank.


\vspace{2mm}{\noindent\bf Related Work and Discussions.}  We are not aware of previous work that explicitly formulates the CDAP framework. Previous work on multi-type resource allocation assumes that items of the same type are interchangeable, and agents have specific preferences, e.g.~{\em Leontief preferences}~\cite{Ghodsi11:Dominant} and threshold preferences~\cite{Huh13:Multiresource}.
CDAPs are more general as agents' preferences are only required to be rankings but not otherwise restricted. 

From the modeling perspective, ignoring the categorial information, CDAPs become standard centralized multi-agent resource allocation problems. 
However, the categorial information opens more possibilities for designing natural allocation mechanisms such as CSAMs. More importantly, we believe that CDAPs provide a natural framework for cross-fertilization of ideas and techniques from other fields of preference representation and aggregation. For example, the combinatorial structure of categorized domains naturally allows agents to use graphical languages (e.g.~{\em CP-nets}~\cite{Boutilier04:CP}) to represent their preferences, which is otherwise hard~\cite{Bouveret09:Conditional}.
Approaches in {\em combinatorial voting}~\cite{Brandt13:Computational} can also be naturally considered in CDAPs. 

Technically, one-sided matching problems are basic CDAPs with one category. Our characterization of serial dictatorships for basic CDAPs are stronger than characterizations of serial dictatorships and similar mechanisms for one-sided matching~\cite{Svensson99:Strategy-proof,Papai00:Strategyproof,Papai00:Strategyproofquotas,Papai01:Strategyproof,Ehlers03:Coalitional,Hatfield09:Strategy-proof}. This is because the category-wise neutrality used in our characterization is weaker than the neutrality used in previous work.  

Our analysis of the worst-case rank efficiency of categorial sequential allocation mechanisms resembles the {\em price of anarchy}~\cite{Koutsoupias99:Worst-case}, which is defined for strategic agents together with a social welfare function that numerically evaluates the quality of outcomes. Our theorem is also related to {\em distortion} in the voting setting~\cite{Procaccia06:Distortion,Boutilier12:Optimal}, which concerns the social welfare loss caused  by agents reporting a ranking instead of a utility function. Nevertheless, our approach is significantly different because we focus on allocation problems for myopic and strategic agents, and we do not assume the existence of agents' cardinal preferences nor a social welfare function, even though our theorem can be easily extended to study worst-case  social welfare loss given a social welfare function, as in Proposition~\ref{prop:optsd} through~\ref{prop:wer}.

Finally, there is a literature in social choice on analyzing the outcomes when agents are myopic~\cite{Brams98:Paradox,Lacy00:Problem,Brams06:Better,Meir14:Local}. These papers focused on voting or cake cutting, while we focus on allocation of indivisible items.
\section{Categorized Domain Allocation Problems}

\begin{dfn} A {\em categorized domain} is composed of $p\geq 1$ categories of indivisible items, denoted by $\{D_1,\ldots,D_p\}$. In a {\em categorized domain allocation problem (CDAP)}, we want to allocate the items to $n$ agents without monetary transfer, such that each agent gets at least one item from each category.

In a {\em basic categorized domain} for $n$ agents, for each $i\leq p$, $|D_i|=n$,  $\md=D_1\times\cdots \times D_p$, and each agent's preferences are represented by a linear order over $\md$. In a {\em basic categorized domain allocation problem (basic CDAP)}, we want to allocate the items to $n$ agents without monetary tranfer, such that every agent gets exactly one item from each category.
\end{dfn} 
In this paper, we focus on basic categorized domains and basic CDAPs for {\em non-sharable} items~\cite{Chevaleyre06:Issues}, that is, each item can only be allocated to one agent.  Therefore, for all $i\leq p$, we write $D_i=\{1,\ldots,n\}$. Each element in $\md$ is called a {\em bundle}. For any $j\leq n$, let $R_j$ denote a linear order over $\md$ and let $P=(R_1,\ldots,R_n)$ denote the agents' {\em (preference) profile}. An {\em allocation} $A$ is a mapping from $\{1,\ldots,n\}$ to $\md$, such that $\bigcup_{j=1}^n[A(j)]_i=D_i$, where for any $j\leq n$ and $i\leq p$, $A(j)$ is the bundle allocated to agent $j$ and $[A(j)]_i$ is the item in category $i$ allocated to agent $j$. 
An  {\em allocation mechanism} $f$ is a mapping that takes a profile as input, and outputs an allocation. We use $f^j(P)$ to denote the bundle allocated to agent $j$  by $f$ for profile $P$.

We now define three axiomatic properties for allocation mechanisms. The first two properties are common in the literature~\cite{Svensson99:Strategy-proof}, and the third is new.

$\bullet$ A direct mechanism $f$ satisfies {\em strategy-proofness} if no agent benefits from misreporting her preferences. That is, for any profile $P$, any agent $j$, and any linear order $R_j'$ over $\md$, $f^j(P)\succ_{R_j}f^j(R_j',R_{-j})$, where $R_{-j}$ is composed of preferences of all agents except agent $j$.

$\bullet$ $f$ satisfies {\em non-bossiness} if no agent is {\em bossy}. An agent is bossy if she can report differently to change the bundles allocated to some other agents without changing her own allocation.  That is, for any profile $P$, any agent $j$, and any linear order $R_j'$ over $\md$, $[f^j(P)=f^j(R_j',R_{-j})]\Rightarrow [f(P)=f(R_j',R_{-j})]$.

$\bullet$ $f$ satisfies {\em category-wise neutrality} if after applying a permutation over the items in a given category the allocation is also permuted in the same way. That is, for any profile $P$, any category $i$, and any permutation $M_i$ over $D_i$, we have $f(M_i(P))=M_i(f(P))$, where for any bundle $\vec d\in\md$, $M_i(\vec d)=(M_i([\vec d]_i),[\vec d]_{-i})$.

When there is only one category, category-wise neutrality degenerates to the traditional neutrality for one-sided matching~\cite{Svensson99:Strategy-proof}. When $p\ge 2$, category-wise neutrality is much weaker than the traditional neutrality.

A {\em serial dictatorship} is defined by a linear order $\mk$ over $\{1,\ldots,n\}$ such that agents choose items in turns according to $\mk$. A truthful agent chooses 
her top-ranked bundle that is still available in each step.

\begin{ex}\label{ex:sd}\rm Let $n=3$ and $p=2$. $\md=\{1,2,3\}\times\{1,2,3\}$. Agents' preferences are as follows. \begin{align*}
R_1=[12\succ 21\succ 32\succ 33\succ 31\succ 22\succ 23\succ 13 \succ 11]\\
R_2=[32\succ 12\succ 21\succ 13 \succ  33 \succ 11\succ 31\succ 23\succ 22]\\
R_3=[13\succ12\succ11\succ 22\succ 32\succ 21\succ 33\succ 31\succ 23]
\end{align*}

Suppose the agents are truthful. Let $\mk=[1\rhd 2\rhd 3]$. In the first round of the serial dictatorship, agent $1$ chooses $12$; in the second round, agent $2$ cannot choose $32$ or $12$ because item $2$ in $D_2$ is unavailable, so she chooses $21$; in the final round, agent $3$ chooses $33$.$\hfill\Box$
\end{ex}

\section{An Axiomatic Characterization}

\begin{thm}\label{thm:sd}For any $p\geq 2$ and $n\geq 2$, an allocation mechanism for a basic categorized domain is strategy-proof, non-bossy, and category-wise neutral if and only if it is a serial dictatorship. Moreover, the three axioms are minimal for characterizing serial dictatorships.
\end{thm}\begin{proof} We first prove four lemmas. The first three lemmas are standard in proving characterizations for serial dictatorships. The last one (Lemma~\ref{lem:notexist}) is new, whose proof is the most involved and heavily relies on the categorial information. 

The first lemma (roughly) says that for all strategy-proof and non-bossy mechanisms $f$ and all profiles $P$, if every agent $j$ reports a different ranking without enlarging the set of bundles ranked above $f^j(P)$ (and she can shuffle the bundles ranked above $f^j(P)$ and she can shuffle the bundles ranked below $f^j(P)$), then the allocation to all agents does not change in the new profile. This resembles {\em (strong) monotonicity} in social choice.

\begin{lem}\label{lem:mono} Let $f$ be a strategy-proof and non-bossy allocation mechanism over a basic categorized domain with $p\geq 2$. For any pair of profiles $P$ and $P'$ such that for all $j\leq n$, $\{\vec d\in\md: \vec d\succ_{R_j'}f^j(P)\}\subseteq \{\vec d\in\md: \vec d\succ_{R_j}f^j(P)\}$, we have $f(P')=f(P)$.
\end{lem}

\begin{proof} We first prove the lemma for the special case where $P$ and $P'$ only differ on one agent's preferences. Let $j$ be an agent with $R_j'\neq R_j$ and $\{\vec d\in\md: \vec d\succ_{R_j'}f^j(P)\}\subseteq \{\vec d\in\md: \vec d\succ_{R_j}f^j(P)\}$. We will prove that $f^j(R_j',R_{-j})=f^j(R_j,R_{-j})$.

Suppose for the sake of contradiction $f^j(R_j',R_{-j})\neq f^j(R_j,R_{-j})$. If $f^j(R_j',R_{-j})\succ_{R_j} f^j(R_j,R_{-j})$ then it means that $f$ is not strategy-proof since $j$ has incentive to report $R_j'$ when her true preferences are $R_j$. If $f^j(R_j,R_{-j})\succ_{R_j} f^j(R_j',R_{-j})$ then $f^j(R_j,R_{-j})\succ_{R_j'} f^j(R_j',R_{-j})$, which means that when agent $j$'s preferences are $R_j'$ she has incentive to report $R_j$,. This again contradicts the assumption that $f$ is strategy-proof. Therefore $f^j(R_j,R_{-j})=f^j(R_j',R_{-j})$. 

By non-bossiness, $f(R_j,R_{-j})=f(R_j',R_{-j})$. The lemma is proved by recursively applying this argument to $j=1,\ldots,n$.
\end{proof}

For any linear order $R$ over $\md$ and any bundle $\vec d\in\md$, we say a linear order $R'$ is a {\em pushup} of $\vec d$ from $R$, if $R'$ can be obtained from $R$ by raising the position of $\vec d$ while keeping the relative positions of other bundles unchanged. The next lemma states that for any strategy-proof and non-bossy mechanism $f$, if an agent reports her preferences differently by only pushing up a bundle $\vec d$, then either the allocation to all agents does not change, or she gets $\vec d$.

\begin{lem}\label{lem:pushup} Let $f$ be a strategy-proof and non-bossy allocation mechanism over a basic categorized domain with $p\geq 2$.  For any profile $P$, any $j\leq n$, any bundle $\vec d$, and any $R_j'$ that is a pushup of $\vec d$ from $R_j$, either (1) $f(R_j',R_{-j}) = f(R)$ or (2) $f^j(R_j',R_{-j}) = \vec d$.
\end{lem}
\begin{proof} We first prove that $f^j(R_j',R_{-j})= f^j(R)$ or $f^j(R_j',R_{-j})= \vec d$. Suppose on the contrary that $f^j(R_j',R_{-j})$ is neither $f^j(R)$ nor $\vec d$. If $f^j(R_j',R_{-j})\succ_{R_j} f^j(R)$, then $f$ is not strategy-proof since when agent $j$'s true preferences are $R_j$ and other agents' preferences are $R_{-j}$, she has incentive to report $R_j'$ to make her allocation better. If $f^j(R)\succ_{R_j}f^j(R_j',R_{-j})$, then since $\vec d\neq f^j(R_j',R_{-j})$, we have $f^j(R)\succ_{R_j'}f^j(R_j',R_{-j})$. In this case when agent $j$'s true preferences are $R_j'$ and other agents' preferences are $R_{-j}$, she has incentive to report $R_j$ to make her allocation better, which means that $f$ is not strategy-proof. Therefore, $f^j(R_j',R_{-j})= f^j(R)$ or $f^j(R_j',R_{-j})= \vec d$. If $f^j(R_j',R_{-j})= f^j(R_j,R_{-j})$, then by non-bossiness $f(R_j',R_{-j})= f(R)$. This completes the proof.
\end{proof}

We next prove that strategy-proofness, non-bossiness, and category-wise neutrality altogether imply {\em Pareto-optimality}, which states that for any profile $P$, there does not exist an allocation $A$ such that all agents prefer their bundles in $A$ to their bundles in $f(P)$, and some of them strictly prefer their bundles in $A$.

\begin{lem}\label{lem:po} For any basic categorized domains with $p\geq 2$,  any strategy-proof, non-bossy, and category-wise neutral allocation mechanism is Pareto optimal.
\end{lem}

\begin{proof} We prove the lemma by contradiction. Let $f$ be a strategy-proof, non-bossy, category-wise neutral, but non-(Pareto optimal) allocation mechanism. Let $P=(R_1,\ldots,R_n)$ denote a profile such that  $f(P)$ is Pareto dominated by an allocation $A$. For any $i\leq m$, let $M_i$  denote the permutation over $D_i$ so that for every $j\leq n$, $[f^j(P)]_i$ is permuted to $[A(j)]_i$. Let $M=(M_1,\ldots, M_m)$. It follows that for all $j\leq n$, $M(f^j(P))=A(j)$.

Let $R_j'$ denote an arbitrary ranking where $A(j)$ is ranked at the top place, and $f^j(P)$ is ranked at the second place if it is different from $A(j)$. Let $R_j^*$ denote an arbitrary ranking where  $f^j(P)$ is ranked at the top place, and $A(j)$ is ranked at the second place if it is different from $f^j(P)$.  Let $P'=(R_1',\ldots,R_n')$ and $P^*=(R_1^*,\ldots,R_n^*)$. $P'$ and $P^*$ are illustrated as follows.

$$P'=\left\{\begin{array}{c}R_1':A^1\succ f^1(P)\succ \text{Others}\\ \vdots\\ R_n':A^n\succ f^n(P)\succ \text{Others}\end{array}\right\}$$

$$P^*=\left\{\begin{array}{c}R_1^*:f^1(P)\succ A^1\succ \text{Others}\\ \vdots\\ R_n^*:f^n(P)\succ A^n\succ \text{Others}\end{array}\right\}$$

Since $A$ Pareto dominates $f(P)$, by Lemma~\ref{lem:mono} we have $f(P')=f(P)$, because for any $j\leq n$, in $R'_j$ the only bundle ranked ahead of $f^j(P)$ is $A(j)$, if it is different from $f^j(P)$, and $A(j)$ is also ranked ahead of $f^j(P)$ in $R_j$. By Lemma~\ref{lem:mono} again we have $f(P^*)=f(P)$. Comparing $M(P')$ and $P^*$, we observe that the only differences are the orderings among $\md\setminus\{A(j),f^j(P)\}$. Applying Lemma~\ref{lem:mono} to $P^*$ and $M(P')$, we have that $f(M(P'))=f(P^*)=f(P)$. However, by category-wise neutrality $f(M(P'))=M(f(P'))=A$, which is a contradiction.
\end{proof}

The next lemma states that for any strategy-proof and non-bossy allocation mechanism $f$, any profile $P$, and any pair of agents $j_1,j_2$, there is no bundle $\vec c$ that only contains items allocated to agent $j_1$ and $j_2$ by $f$, such that both $j_1$ and $j_2$ prefer $\vec c$ to their bundles allocated by $f$.

    \renewcommand{\arraystretch}{1.3}
    \setlength\extrarowheight{2pt}
\begin{table*}[htp]
\tbl{Proof steps for Lemma~\ref{lem:notexist}.\label{tab:proof}}{
\begin{tabular}{ccc}
\begin{tabular}{@{}c@{}}$\boxit{\begin{array}{l}\hat R_{j_1}: \vec c\succ\boxit{\vec a}\succ \vec d\succ \vec b\succ \text{others}\\
\hat R_{j_2}:\vec c\succ \boxit{\vec b}\succ \vec a\succ \vec d\succ\text{others}\\
\text{Other $j$}: f^j(P)\succ\text{others}\end{array}}$\\
Step 1\end{tabular}

&\begin{tabular}{@{}c@{}}\boxit{\begin{array}{l}\hat R_{j_1}: \vec c\succ \boxit{\vec a}\succ \vec d\succ \vec b\succ \text{others}\\
\bar R_{j_2}:\vec c\succ \vec a\succ \boxit{\vec b}\succ\vec d\succ \text{others}\\
\text{Other $j$}: f^j(P)\succ\text{others}\end{array}}\\ Step 2\end{tabular}

&\begin{tabular}{@{}c@{}}\boxit{\begin{array}{l}\bar R_{j_1}: \vec c\succ \boxit{\vec b}\succ \vec a\succ \vec d\succ \text{others}\\
\bar R_{j_2}:\vec c\succ \boxit{\vec a}\succ {\vec b}\succ \vec d\succ \text{others}\\
\text{Other $j$}: f^j(P)\succ\text{others}\end{array}}\\ Step 3\end{tabular}\\
\\
\begin{tabular}{@{}c@{}}\boxit{\begin{array}{l}\bar R_{j_1}: \vec c\succ \boxit{\vec b}\succ \vec a\succ \vec d\succ \text{others}\\
{\mathring R}_{j_2}:\vec c\succ \boxit{\vec a}\succ \vec d\succ {\vec b}\succ \text{others}\\
\text{Other $j$}: f^j(P)\succ\text{others}\end{array}}\\ Step 4\end{tabular}

& \begin{tabular}{@{}c@{}}\boxit{\begin{array}{l}{\mathring R}_{j_1}: \vec c\succ \vec a\succ \boxit{\vec b}\succ \vec d\succ \text{others}\\
{\mathring R}_{j_2}:\vec c\succ \boxit{\vec a}\succ \vec d\succ {\vec b}\succ \text{others}\\
\text{Other $j$}: f^j(P)\succ\text{others}\end{array}}\\ Step 5\end{tabular}

&\begin{tabular}{@{}c@{}}\boxit{\begin{array}{l}{\mathring R}_{j_1}: \vec c\succ \boxit{\vec a}\succ \vec b\succ \vec d\succ \text{others}\\
{\bar R}_{j_2}:\vec c\succ \vec a\succ \boxit{\vec b}\succ \vec d\succ\text{others}\\
\text{Other $j$}: f^j(P)\succ\text{others}\end{array}}\\ Step 6\end{tabular}
\end{tabular}}
\end{table*}

\begin{lem}\label{lem:notexist} Let $f$ be a strategy-proof and non-bossy allocation mechanism over a basic categorized domain with $p\geq 2$.  For any profile $P$ and any $j_1\neq j_2\leq n$, let $\vec a = f^{j_1}(P)$ and $\vec b = f^{j_2}(P)$, there does not exist $\vec c\in \{a_1,b_1\}\times \{a_2, b_2\}\times\cdots\times \{ a_p,b_p\}$ such that $\vec c \succ_{R_{j_1}}\vec a$ and $\vec c \succ_{R_{j_2}}\vec b$, where $a_i$ is the $i$-th component of $\vec a$.
\end{lem}

\begin{proof}
Suppose for the sake of contradiction that such a bundle $\vec c$ exists. Let $\vec d$ denote the bundle such that $\vec c\cup\vec d=\vec a\cup \vec b$. More precisely, for all $i\leq m$, $\{c_i,d_i\} = \{a_i,b_i\}$. For example, if $\vec a  = 1213$, $\vec b = 2431$, and $\vec c = 1211$, then $\vec d = 2433$.

The rest of the proof derives a contradiction by proving the  series of observations  illustrated in Table~\ref{tab:proof}. In each step, we prove that the boxed bundles are allocated to agent $j_1$ and agent $j_2$ respectively, and all other agents get their top-ranked bundles.

\noindent{\bf Step~1.} Let $\hat R_{j_1}=[\vec c\succ \vec a\succ \vec d\succ \vec b \succ \text{others}]$, $\hat R_{j_2}=[\vec c\succ \vec b\succ \vec a \succ \vec d\succ \text{others}]$, where ``others'' represents an arbitrary linear order over the remaining bundles, and for any $j\neq j_1,j_2$, let $\hat R_{j}=[f^j(P)\succ\text{others}$]. By Lemma~\ref{lem:mono}, $f(\hat P)=f(P)$. 

\noindent{\bf Step~2.} Let  $\bar R_{j_2}=[\vec c\succ \vec a\succ {\vec b}\succ \vec d\succ \text{others}]$ be a pushup of $\vec a$ from $\hat R_{j_2}$. We will prove that $f(\bar R_{j_2}, \hat R_{-j_2})=f(\hat P)=f(P)$. Since $\bar R_{j_2}$ is a pushup of $\vec a$ from $\hat R_{j_2}$, by Lemma~\ref{lem:pushup}, $f^{j_2}(\bar R_{j_2}, \hat R_{-j_2})$ is either $\vec a$ or $\vec b$.  We now show that the former case is impossible. Suppose for the sake of contradiction $f^{j_2}(\bar R_{j_2}, \hat R_{-j_2})=\vec a$, then $f^{j_1}(\bar R_{j_2}, \hat R_{-j_2})$ cannot be $\vec c$, $\vec a$, or $\vec d$ since otherwise some item will be allocated twice. This means that  $f(\bar R_{j_2}, \hat R_{-j_2})$ is Pareto dominated by the allocation where $j_1$ gets $\vec d$, $j_2$ gets $\vec c$, and all other agents get their top-ranked bundles. This contradicts the Pareto-optimality of $f$ (Lemma~\ref{lem:po}). Hence $f^{j_2}(\bar R_{j_2}, \hat R_{-j_2})=\vec b=f^{j_2}(\hat P)$. By non-bossiness we have $f(\bar R_{j_2}, \hat R_{-j_2})=f(\hat P)=f(P)$.

\noindent{\bf Step~3.} Let ${\bar R}_{j_1}=[\vec c\succ {\vec b}\succ \vec a\succ \vec d\succ \text{others}]$ be a pushup of $\vec b$ from $\hat R_{j_1}$. We will prove that in $f({\bar R}_{j_1},\bar R_{j_2}, \hat R_{-\{j_1,j_2\}})$, $j_1$ gets $\vec b$, $j_2$ gets $\vec a$, and all other agents get the same items as in $f(P)$. Since ${\bar R}_{j_1}$ is a pushup of $\vec b$ from $\hat R_{j_1}$, by Lemma~\ref{lem:pushup}, $f^{j_1}({\bar R}_{j_1},\bar R_{j_2}, \hat R_{-\{j_1,j_2\}})$ is either $\vec a$ or $\vec b$. We now show that the former case is impossible. Suppose for the sake of contradiction that $f^{j_1}({\bar R}_{j_1},\bar R_{j_2}, \hat R_{-\{j_1,j_2\}})=\vec a$. By non-bossiness, $f^{j_2}({\bar R}_{j_1},\bar R_{j_2}, \hat R_{-\{j_1,j_2\}})=\vec b$. This means that $f({\bar R}_{j_1},\bar R_{j_2}, \hat R_{-\{j_1,j_2\}})$ is Pareto-dominated by the allocation where $j_1$ gets $\vec b$, $j_2$ gets $\vec a$, and all other agents get their top-ranked bundles. This contradicts the Pareto-optimality of $f$ (Lemma~\ref{lem:po}).

\noindent{\bf Step~4.} Let ${\mathring R}_{j_2}=[\vec c\succ {\vec a}\succ \vec d\succ {\vec b}\succ \text{others}]$ be a pushup of $\vec d$ from $\bar R_{j_2}$. By Lemma~\ref{lem:mono}, $f({\bar R}_{j_1},{\mathring R}_{j_2}, \hat R_{-\{j_1,j_2\}})=f({\bar R}_{j_1},\bar R_{j_2}, \hat R_{-\{j_1,j_2\}})$. 

\noindent{\bf Step~5.} Let ${\mathring R}_{j_1}=[\vec c\succ {\vec a}\succ \vec b\succ \vec d\succ \text{others}]$ be a pushup of $\vec a$ from $\bar R_{j_1}$. We will prove that $f({\mathring R}_{j_1},{\mathring R}_{j_2}, \hat R_{-\{j_1,j_2\}})=f({\bar R}_{j_1},{\mathring R}_{j_2}, \hat R_{-\{j_1,j_2\}})$. Since ${\mathring R}_{j_1}$ is a pushup of $\vec a$ from ${\bar R}_{j_1}$, by Lemma~\ref{lem:pushup}, $f^{j_1}({\mathring R}_{j_1}, {\mathring R}_{j_2}, \hat R_{-\{j_1,j_2\}})$ is either $\vec a$ or $\vec b$. We now show that the former case is impossible. Suppose for the sake of contradiction that $f^{j_1}({\mathring R}_{j_1}, {\mathring R}_{j_2}, \hat R_{-\{j_1,j_2\}})=\vec a$. Then in $f({\mathring R}_{j_1},{\mathring R}_{j_2}, \hat R_{-\{j_1,j_2\}})$, agent $j_2$ cannot get $\vec c$, $\vec a$, or $\vec d$, which means that $f({\mathring R}_{j_1},{\mathring R}_{j_2}, \hat R_{-\{j_1,j_2\}})$ is Pareto-dominated by the allocation where $j_1$ gets $\vec c$, $j_2$ gets $\vec d$, and all other agents get their top-ranked bundles. This contradicts the Pareto-optimality of $f$. Hence, $f^{j_1}({\mathring R}_{j_1},{\mathring R}_{j_2}, \hat R_{-\{j_1,j_2\}})=\vec b$. By non-bossiness  $f({\mathring R}_{j_1},{\mathring R}_{j_2}, \hat R_{-\{j_1,j_2\}})=f({\bar R}_{j_1},{\mathring R}_{j_2}, \hat R_{-\{j_1,j_2\}})$.

\noindent{\bf Step~6.} We note that ${\mathring R}_{j_1}$ is a pushup of $\vec b$ from $\hat R_{j_1}$ (and $\vec b$ is still below $\vec a$). By Lemma~\ref{lem:mono}, $f({\mathring R}_{j_1},\bar {R}_{j_2}, \hat R_{-\{j_1,j_2\}})=f(\hat{ R}_{j_1},\bar {R}_{j_2}, \hat R_{-\{j_1,j_2\}})$. We note that the right hand side is the profile in Step~2.

{\noindent\bf Contradiction.} Finally, the observations in Step~5 and Step~6 imply that when agents' preferences are as in Step~6, agent $j_2$ has incentive to report ${\mathring R}_{j_2}$ in Step~5 to improve the bundle allocated to her (from $\vec b$ to $\vec a$). This contradicts the strategy-proofness of $f$ and completes the proof of Lemma~\ref{lem:notexist}. 
\end{proof}

It is easy to check that any serial dictatorship satisfies strategy-proofness, non-bossiness and category-wise neutrality. We now prove that any mechanism satisfying the three axioms must be a serial dictatorship. Let $R^*$ be a linear order over $\md$ that satisfies the following conditions:

$\bullet$  $(1,\ldots, 1)\succ (2,\ldots, 2)\succ\cdots\succ (n,\ldots, n)$.

$\bullet$  For any $j<n$, the bundles ranked between $(j,\ldots,j)$ and $(j+1,\ldots,j+1)$ are those satisfying the following two conditions: 1) at least one component is $j$, and 2) all components are in $\{j,j+1,\ldots,n\}$. Let $B_j$ denote these bundles. That is, $B_j\subseteq \md$ and $B_j=\{\vec d:\forall l, d_l\geq j \text{ and }\exists l', d_{l'}= j\}$.

$\bullet$  For any $j$ and any $\vec d, \vec e\in B_j$, if the number of $j$'s in $\vec d$ is strictly larger than the number of $j$'s in $\vec e$, then $\vec d\succ \vec e$. 

The next claim states that $f$ agrees with a serial dictatorship on a specific profile.
\begin{claim}\label{claim:sd}Let $P^*=(R^*,\ldots,R^*)$. For any $l\leq n$, there exists $j_l\leq n$ such that $f^{j_l}(P^*)=(l,\ldots, l)$.\end{claim}

\begin{proof} The claim is proved by induction on $l$. When $l=1$, for the sake of contradiction suppose there is no $j_l$ with $f^{j_l}(P^*)=(1,\ldots,1)$. Then there exist a pair of agents $j$ and $j'$ such that both $\vec a = f^{j}(P^*)$ and $\vec b = f^{j'}(P^*)$ contain $1$ in at least one category. 

Let $\vec c$ be the bundle obtained from $\vec a$ by replacing items in categories where $\vec b$ takes $1$ to $1$. More precisely, we let $\vec c = (c_1,\ldots, c_p)$, where 
$$c_i=\left\{\begin{array}{cl}1&\text{if }a_i=1\text{ or }b_i=1\\a_i&\text{otherwise}\end{array}\right.$$

It follows that in $R^*$, $\vec c\succ_{R^*} \vec a$ and $\vec c\succ_{R^*} \vec b$ since the number of $1$'s in $\vec c$  is strictly larger than the number of $1$'s in $\vec a$ or $\vec b$. By Lemma~\ref{lem:notexist}, this contradicts the assumption that $f$ is strategy-proof and non-bossy. Hence there exists $j_1\leq n$ with $f^{j_1}(P^*)=(1,\ldots,1)$.

Suppose the claim is true for $l\leq l'$. We next prove that there exists $j_{l'+1}$ such that $f^{j_{l'+1}}(P^*)=(l'+1,\ldots,l'+1)$. This follows after a similar reasoning to the $l=1$ case. Formally, suppose for the sake of contradiction there does not exist such a $j_{l'+1}$.  Then, there exist two agents who get $\vec a$ and $\vec b$ in $f(P^*)$ such that both $\vec a$ and $\vec b$ contain $l'+1$ in at least one category. By the induction hypothesis, items $\{1,\ldots,l'\}$ in all categories have been allocated, which means that all components of $\vec a$ and $\vec b$ are at least as large as $l'+1$. Let $\vec c$ be the bundle obtained from $\vec a$ by replacing items in all categories where $\vec b$ takes $l'+1$ to $l'+1$. We have $\vec c\succ_{R^*} \vec a$ and $\vec c\succ_{R^*}\vec b$, leading to a contradiction by Lemma~\ref{lem:notexist}. Therefore, the claim holds for $l=l'+1$. This completes the proof of Claim~\ref{claim:sd}.
\end{proof}

W.l.o.g.~we let $j_1=1$, $j_2=2$, $\ldots$, $j_n=n$ denote the agents in Claim~\ref{claim:sd}. For any profile $P'=(R_1',\ldots,R_n')$, we define $n$ bundles as follows. Let $\vec {d^1}$ denote the top-ranked bundle in $R_{1}'$, and for any $l\geq 2$, let $\vec {d^l}$ denote agent $l$'s top-ranked available bundle given that items in $\vec {d^1},\ldots,\vec{d^{l-1}}$ have already been allocated. That is, $\vec{d^l}$ is the most preferred bundle in $\{\vec d:\forall l'<l, \vec d\cap \vec {d^{l'}}=\emptyset\}$ according to $R_{l}'$. In other words, $\vec d^1,\ldots, \vec d^n$ are the bundles allocated to agents $1$ through $n$ by the serial dictatorship $1\rhd 2\cdots\rhd n$. We next prove that this is exactly the allocation by $f$.

For any $i\leq m$, we define a category-wise permutation $M_i$ such that for all $l\leq n$, $M_i(l)=[\vec {d^l}]_i$, where we recall that $[\vec {d^l}]_i$ is the item in the $i$-th category in $\vec{d^l}$. Let $M=(M_1,\ldots,M_m)$. It follows that for all $l\leq n$, $M(l,\ldots,l)=\vec {d^l}$. By category-wise neutrality and Claim~\ref{claim:sd}, in $f(M(P^*))$ agent $l$ gets $M(f^l(P^*))=\vec{d^l}$. 

Comparing $M(P^*)$ to $P'$, we notice that for all $l\leq n$ and all bundles $\vec e$, if $ \vec {d^l}\succ_{M(R^*)}\vec e$ then $\vec {d^l}\succ_{R_l'} \vec e$.  This is because if there exists $\vec e$ such that  $ \vec {d^l}\succ_{M(R^*)}\vec e$ but $\vec e\succ_{R_l'} \vec {d^l}$, then $\vec e$ is still available after $\{\vec {d^1},\ldots,\vec {d^{l-1}}\}$ have been allocated, and $\vec e$ is ranked higher than $\vec {d^l}$ in $R_l'$. This contradicts the selection of $\vec {d^l}$. By Lemma~\ref{lem:mono}, $f(P')=f(M(P^*))=M(f(P^*))$, which proves that $f$ is the serial dictatorship w.r.t.~the order $1\rhd 2\rhd\cdots\rhd n$. 

Next, we show that strategy-proofness, non-bossiness, and category-wise neutrality are a minimal set of properties that characterize serial dictatorships.

\begin{paragraph}{\bf Strategy-proofness is necessary} Consider the allocation mechanism that maximizes the social welfare w.r.t.~the following utility functions. For any $i\leq n^p$ and $j\leq n$, the bundle ranked at the $i$-th position in agent $j$'s preferences gets $(n^p-i)(1+(\frac{1}{2n^p})^j)$ points.\footnote{The $(\frac{1}{2n^p})^j$ terms in the utility functions are only used to avoid ties in allocations. In fact, any utility functions where there are no ties satisfy non-bossiness and category-wise neutrality, but some of them are equivalent to serial dictatorships, which are the cases we want to avoid in our proof.} It is not hard to check that for any pair of different allocations, the social welfares are different. It follows that this allocation mechanism satisfies non-bossiness. This is because if agent $j$'s allocation is the same when only she reports differently, then the set of items left to the other agents is the same, which means that the allocation to the other agents by the mechanism is the same. Since the utility of a bundle only depends on its position in the agents' preferences rather than the name of the bundle, the allocation mechanism satisfies category-wise neutrality. This mechanism is not a serial dictatorship. To see this, consider the basic categorized domain with $p=n=2$, $R_1'=[11\succ 12\succ 22\succ 21]$, and $R_2'=[12\succ 21\succ 11\succ 22]$. A serial dictatorship will either give $11$ to agent $1$ and give $22$ to agent $2$, or give $21$ to agent $1$ and give $12$ to agent $2$, but the allocation that maximizes social welfare w.r.t.~the utility function described above is to give $12$ to agent $1$ and give $21$ to agent $2$.
\end{paragraph}

\begin{paragraph}{\noindent\bf Non-bossiness is necessary} Consider the following ``conditional serial  dictatorship'': agent $1$ chooses  her favorite bundle in the first round, and the order over the remaining agents $\{2,\ldots,n\}$ depends on agent $1$'s preferences in the following way: if the first component of agent $1$'s second-ranked bundle is the same as the first component of her top choice, then the order over the rest of agents is $2\rhd 3\rhd\cdots\rhd n$; otherwise it is $n\rhd n-1\rhd\cdots\rhd 2$. It is not hard to verify that this mechanism satisfies strategy-proofness and category-wise neutrality, and is not a serial dictatorship (where the order must be fixed before seeing the profile).\end{paragraph}

\begin{paragraph}{\bf Category-wise neutrality is necessary} Consider the following ``conditional serial  dictatorship'': agent $1$ chooses her favorite bundle in the first round, and the order over agents $\{2,\ldots,n\}$ depends on the allocation to agent $1$ in the following way: if agent $1$ gets $(1,\ldots, 1)$, then the order over the rest of agents is $2\rhd 3\rhd\cdots\rhd n$; otherwise it is $n\rhd n-1\rhd\cdots\rhd 2$. It is not hard to verify that this mechanism satisfies strategy-proofness and non-bossiness, and is not a serial dictatorship.
\end{paragraph}
\end{proof}

{\bf Remarks.} The theorem is somewhat negative, meaning that we have to sacrifice one of strategy-proofness, category-wise neutrality, or non-bossiness. Among the three axiomatic properties, we feel that non-bossiness is the least natural one.

\section{Categorial Sequential Allocation Mechanisms}\label{sec:sequential} 

Given a linear order $\mo$ over $\{1,\ldots,n\}\times\{ 1,\ldots, p\}$, the {\em categorial sequential allocation mechanism (CSAM)} $f_\mo$ allocates the items in $np$ steps as illustrated in Protocol~\ref{alg:protocol}. In each step $t$, suppose the $t$-th element in $\mo$ is $(j,i)$, (equivalently, $t=\mo^{-1}(j,i)$). Agent $j$ is called the {\em active agent} in step $t$ and she chooses an item $d_{j,i}$ that is still available from $D_i$. Then, $d_{j,i}$ is broadcast to all agents and we move on to the next step. 
\renewcommand*{\algorithmcfname}{Protocol}
\IncMargin{1em}
\begin{algorithm}[!htp]
\SetAlgoLined
\LinesNumbered
\caption{Categorial sequential allocation mechanism (CSAM) $f_\mo$.\label{alg:protocol}}\DontPrintSemicolon

\Indm
\KwIn{An order $\mo$  over $\{1,\ldots,n\}\times\{ 1,\ldots, p\}$.}
\Indp
Broadcast $\mo$ to all agents.\label{step:broadcast}\;
\For{$t=1$ to $np$\label{step:start}}{Let $(j,i)$ be the $t$-th element in $\mo$.\;
Agent $j$ chooses an available item $d_{j,i}\in D_i$.\;
Broadcast $d_{j,i}$ to all agents.\label{step:end}}
\end{algorithm}

In CSAMs, in each step the active agent must choose an item from the designated category. Hence, CSAMs are different from sequential allocation protocols~\cite{Bouveret11:General} and the draft mechanism~\cite{Budish12:Multi}, where in each step the active agent can choose any available item from any category.

\begin{ex}\label{ex:csm}\rm The serial dictatorship w.r.t.~$\mk=[j_1\rhd \cdots\rhd j_n]$ is a CSAM w.r.t.~$(j_1,1)\rhd (j_1,2)\rhd\cdots\rhd(j_1,p)\rhd\cdots\rhd(j_n,1)\rhd (j_n,2)\rhd\cdots\rhd(j_n,p)$. 

For any even number $p$, given any linear order $\mk=[j_1\rhd \cdots\rhd j_n]$ over the agents, we define the {\em balanced CSAM} to be the mechanism where agents choose items in $p$ phases, such  that for each $i\leq p$, in phase $i$ all agents choose from $D_i$ w.r.t.~$\mk$ if $i$ is odd, and w.r.t.~inverse $K$ if $i$ is even.


For example, when $n=3$, $p=2$, and $\mk=[1\rhd 2\rhd 3]$, the balanced CSAM uses the order $(1,1)\rhd (2,1)\rhd (3,1)\rhd (3,2)\rhd (2,2)\rhd (1,2)$. $\hfill\Box$
\end{ex}


Similar to sequential allocations~\cite{Bouveret11:General}, CSAMs can be implemented in a distributed manner. Communication cost for CSAMs is much lower than for direct mechanisms, where agents report their preferences in full to the center, which requires $\Theta(n^{p}p\log n)$ bits per agent, and thus the total communication cost is $\Theta(n^{p+1}p\log n)$. For CSAMs, the total communication cost of Protocol~\ref{alg:protocol} is $\Theta(n^2p\log n+ np(n\log n))=\Theta(n^2p\log np)$, which has  $\Theta(n^{p-2}\cdot \frac{\log n}{\log n+\log p})$ multiplicative saving. In light of this, CSAMs preserve more privacy as well.

We will consider two types of myopic agents. 
For any $1\leq i\leq p$, we let $D_{i,t}$ denote the set of available items in $D_{i}$ at the beginning of round $t$.

\begin{itemize}
\item {\em Optimistic} agents. An optimistic agent chooses the item in her top-ranked bundle that is still available, given the items she chose in previous steps.
\item {\em Pessimistic} agents. A pessimistic agent $j$ in round $t$ chooses an item $d_{j,i}$ from $D_{i,t}$,   such that for all $d_i'\in D_{i,t}$ with $d_i'\neq d_{j,i}$, agent $j$ prefers the worst available bundle whose $i$-th component is $d_{j,i}$ to the worst available bundle whose $i$-th component is $d_i'$.

\end{itemize}
In this paper, we assume that whether an agent is optimistic or pessimistic is  fixed before applying a CSAM. 
\begin{ex}\label{ex:sa}\rm 
Let $n=3$, $p=2$. Consider the same profile as in Example~\ref{ex:sd}, which can be simplified as follows.
$$\begin{array}{rl}
\text{Agent 1  (optimistic):}&12\succ 21\succ \text{\rm others}\succ 11\\
\text{Agent 2  (optimistic):}&32\succ\text{\rm others}\succ 22\\
\text{Agent 3  (pessimistic):}&13\succ \text{\rm others}\succ 33\succ 31\succ 23
\end{array}$$

Let $\mo=[(1,1)\rhd (2,2)\rhd (3,1)\rhd (3,2)\rhd (2,1)\rhd (1,2)]$. Suppose agent 1 and agent 2 are optimistic and agent 3 is pessimistic. When $t=1$, agent $1$ (optimistic) chooses item $1$ from $D_1$. When $t=2$, item  $32$ is the top-ranked available bundle for agent $2$ (optimistic), so she chooses $2$ from $D_2$. When $t=3$, the available bundles are $\{2,3\}\times\{1,3\}$. If agent $3$ chooses $2$ from $D_1$, then the worst-case available bundle is $23$, and if agent $3$ chooses $3$ from $D_1$, then the worst-case available bundle is $31$. Since agent $3$ prefers $31$ to $23$, she chooses $3$ from $D_1$. When $t=4$,  agent $3$ chooses $3$ from $D_2$. When $t=5$, agent $2$ choses $2$ from $D_1$ and when $t=6$, agent $1$ choses $1$ from $D_2$. Finally, agent $1$ gets $11$, agent $2$ gets $22$, and agent $3$ gets $33$. $\hfill\Box$

\end{ex}


\section{Rank Efficiency of CSAMs for Myopic Agents}\label{sec:eff}
In this section, we 
focus on characterizing the {\em rank efficiency} of CSAMs measured by agents' ranks of the bundles they receive.\footnote{This is different from the ordinal efficiency for randomized allocation mechanisms~\cite{Bogomolnaia01:New}.}  For any linear order $R$ over $\md$ and any bundle $\vec d$, we let $\rank(R,\vec d)$ denote the {rank} of $\vec d$ in $R$, such that the highest position has rank $1$ and the lowest position has rank $n^p$. 
Given a CSAM $f_\mo$, we introduce the following notation for any $j\leq n$.

$\bullet$  Let $\mo_j$ denote the linear  order over the categories $\{1,\ldots,p\}$ according to which agent $j$ chooses items from $\mo$. 

$\bullet$  For any $i\leq p$, let $k_{j,i}$ denote the number of  items in $D_i$ that are still available right before agent $j$ chooses from $D_i$. Formally, $k_{j,i}=1+|\{(j',i):(j,i)\rhd_\mo (j',i)\}|$. 

$\bullet$  Let $K_j$ denote the smallest index in $\mo_j$ such that no agent can ``interrupt'' agent $j$ from choosing all items in her top-ranked bundle that is available in round $(j,\mo_j(K_j))$. Formally, $K_j$ is the smallest number such that for any $l$ with $K_j<l\leq p$, between the round when agent $j$ chooses an item from category $\mo_j(K_j)$ and the round when agent $j$ chooses an item from category $\mo_j(l)$, no agent chooses an item from category $\mo_j(l)$. We note that $K_j$ is defined only by $\mo$ and is thus independent of agents' preferences.

\begin{ex}\label{ex:o}\rm 
Let $\mo^*=[(1,1)\rhd (1,2)\rhd (1,3)\rhd (2,1)\rhd (2,2)\rhd (2,3)\rhd (3,1)\rhd (3,2)\rhd (3,3)]$. That is, $f_{\mo^*}$ is a serial dictatorship. Then $\mo_1^*=\mo_2^*=\mo_3^*=1\rhd 2\rhd 3$. $K_1=K_2=K_3=1$. $k_{1,1}=k_{1,2}=k_{1,3}=3$, $k_{2,1}=k_{2,2}=k_{2,3}=2$, $k_{3,1}=k_{3,2}=k_{3,3}=1$.

Let $\mo$ be the order in Example~\ref{ex:sa}, that is, $\mo =[(1,1)\rhd (2,2)\rhd (3,1)\rhd (3,2)\rhd (2,1)\rhd (1,2)]$. 

$\mo_1=1\rhd 2$. $K_1=2$ since $(2,2)$ is between $(1,1)$ and $(1,2)$ in $\mo$. $k_{1,1}=3$, $k_{1,2}=1$.

$\mo_2=2\rhd 1$. $K_2=2$ since $(3,1)$  is between $(2,2)$ and $(2,1)$. $k_{2,1}=1$, $k_{2,2}=3$.

$\mo_3=1\rhd 2$. $K_3=1$ since between $(3,1)$ and $(3,2)$ in $\mo$, no agent chooses an item from $D_2$. $k_{3,1}=k_{3,2}=2$.$\hfill\Box$
\end{ex}

\begin{prop}\label{prop:worstcase} For any CSAM $f_\mo$, any combination of optimistic and pessimistic agents, any $j\leq n$, and any profile:
\begin{itemize}
\item {\bf Upper bound for optimistic agents:} if $j$ is optimistic, then the rank of the bundle allocated to her is at most $n^p+1-\prod_{l=K_j}^pk_{j,\mo_j(l)}$.

\item {\bf Upper bound for pessimistic agents:} if $j$ is pessimistic, then the rank of the bundle  allocated to her is at most $n^p-\sum_{l=1}^p(k_{j,\mo_j(l)}-1)$.
\end{itemize}
\end{prop}
\begin{proof} Equivalently, we prove that for any optimistic agent, the bundle allocated to her is ranked no lower than the $(\prod_{l=K_j}^pk_{j,\mo_j(l)})$-th position from the bottom, and for any pessimistic agent, the bundle allocated to her is ranked no lower than the $(1+\sum_{l=1}^p(k_{j,\mo_j(l)}-1))$-th position from the bottom.

W.l.o.g.~let $\mo_j=1\rhd 2\rhd \cdots \rhd p$. That is, agent $j$ chooses items from categories $1,\ldots,p$ in sequence in the sequential allocation. This means that in this proof, for any $l\leq p$, $\mo_j(l)=l$. We first prove the proposition for an optimistic agent $j$.  In the beginning of round $t_j=\mo^{-1}(j,{K_j})$ in Algorithm~\ref{alg:protocol}, agent $j$ has already chosen items from $D_1,\ldots,D_{K_j-1}$, and is ready to choose an item from $D_{K_j}$. We recall that  $D_{l,t}$ is the set of remaining items in $D_{l}$ at the beginning of round $t$.  By definition, $k_{j,l}=|D_{l,t_j}|$. Let $(d_{j,1},\ldots,d_{j,{K_j-1}})\in D_1\times\cdots\times D_{K_j-1}$ denote the items agent $j$ has chosen in previous rounds. It follows that at the beginning of the round $t_j$, the following $\prod_{l=K_j}^pk_{j, l}$ bundles are available for agent $j$: 
$$\md_{j}=(d_{j,1},\ldots,d_{j,{K_j-1}})\times \prod_{l=K_j}^pD_{l,t_j}$$

We now show that an optimistic agent $j$ is guaranteed to obtain her top-ranked bundle in $\md_{j}$. Intuitively this holds because by the definition of $K_j$, for any $l\geq K_j$, when it is  agent $j$'s round to choose an item from $D_{l}$, the $l$-th component of her top-ranked bundle in $\md_{j}$ is always available. Formally, let $\vec d_j = (d_{j,1},\ldots, d_{j,p})$ denote agent $j$'s top-ranked bundle in $\md_{j}$. We prove that agent $j$ will choose $d_{j,l}$ from $D_l$ in round $\mo^{-1}(j,l)$ by induction on $l$.  The base case  $l=K_j$ is straightforward.
Suppose she has chosen $d_{K_j}, d_{K_j+1},\ldots,d_{l'}$ for some $l'\ge K_j$. Then in round $\mo^{-1}(j,{l'+1})$ when agent $j$ is about to choose an item from $D_{l'+1}$, the following bundles are available: 
$$(d_{j,1},\ldots,d_{j,l'})\times \prod_{l=l'+1}^pD_{l,t_j}$$

This is because by the induction hypothesis $(d_{j,1},\ldots,d_{j,l'})$ have been chosen by agent $j$ in previous rounds. Then, by the definition of $K_j$, for any $l\geq l'+1$ no agent chooses an item from $D_l$ between round $t_j=\mo^{-1}(j,K_j)$ and round $\mo^{-1}(j,l')$. Hence the remaining items in $D_l$ are still the same as those in round $t_j$. This means that $\vec d_j\in (d_{j,1},\ldots,d_{j,l'})\times \prod_{l=l'+1}^pD_{l,t_j}\subseteq \md_{j}$. Therefore, $\vec d_j$ is still agent $j$'s top-ranked available bundle in the beginning of round $\mo^{-1}(j,l')$, when she is about to choose an item from $D_{l'+1}$. Hence agent $j$ will choose $d_{j,l'+1}$. This proves the claim for $l=l'+1$, which means that it holds for all $l\leq p$. Therefore, agent $j$ is allocated $\vec d_j$ by the sequential allocation protocol. We note that $|\md_{j}|=\prod_{l=K_j}^pk_{j,l}$. This proves the proposition for optimistic agents.

We next prove the proposition for an pessimistic agent $j$. Let $\vec d_j=(d_{j,1},\ldots,d_{j,p})$ denote her allocation by the sequential allocation protocol. Since agent $j$ is pessimistic, for any $1\leq l\leq p$, in round $t^*=\mo^{-1}(j,l)$ agent $j$ chose $d_{j,l}$ from $D_{l,t^*}$, we must have that for all $d_{l}'\in D_{l,t^*}$ with $d_{l}'\neq d_{j,l}$, there exists a bundle $(d_{j,1},\ldots, d_{j,l-1}, d_{l}',\ldots, d_{p}')$ that is ranked below $\vec d_j$. These bundles are all different and the number of all such bundles is $\sum_{l=1}^p(k_{j,l}-1)$, which proves the proposition for pessimistic agents.
\end{proof}

We note that Proposition~\ref{prop:worstcase} works for any combination of optimistic and pessimistic agents, which is much more general than the setting with all-optimistic agents and the setting with all-pessimistic agents. In addition, once the CSAM and the properties of the agents (that is, whether each agent is optimistic or pessimistic) is given, the bounds hold for all preference profiles.

Our main theorem in this section states that, surprisingly, for all combinations of optimistic and pessimistic agents, all upper bounds described in Proposition~\ref{prop:worstcase} can be matched in a same profile. Even more surprisingly, for the same profile there exists an allocation where almost all agents get their top-ranked bundle, and the only agent who may not get her top-ranked bundle gets her second-ranked bundle. Therefore, the theorem not only provides a worst-case analysis in the absolute sense in that all upper bounds in Proposition~\ref{prop:worstcase} are matched in the same profile, but also in the comparative sense w.r.t.~the optimal allocation of the profile. 
\begin{thm}\label{thm:matching} For any CSAM $f_\mo$ and any combination of optimistic and pessimistic agents, there exists a profile $P$ such that for all $j\leq n$:
\begin{enumerate}
\item if agent $j$ is optimistic, then the rank of the bundle allocated to her is $n^p+1-\prod_{l=K_j}^pk_{j,\mo_j(l)}$;
\item if agent $j$ is pessimistic, then  the rank of the bundle  allocated to her is $n^p-\sum_{l=1}^p(k_{j,\mo_j(l)}-1)$;
\item there exists an allocation where at least $n-1$ agents get their top-ranked bundles, and the remaining agent gets her top-ranked or second-ranked bundle. 
\end{enumerate}
\end{thm}
The proof is quite involved and can be found in the supplementary material.

\begin{ex}\rm  The profile in Example~\ref{ex:sa} is an example of the profile guaranteed by Theorem~\ref{thm:matching}: agent $1$ (optimistic) gets her bottom bundle ($K_1=2$ and $k_{1,2}=1$), agent $2$ (optimistic) gets her bottom bundle ($K_2=2$ and $k_{2,1}=1$), and agent $3$ (pessimistic) gets her third bundle ($k_{3,1}=k_{3,2}=2$). Moreover, there exists an allocation where agent $2$ and agent $3$ get their top bundles and agent $1$ gets her second bundle.
$\hfill\Box$
\end{ex}

Theorem~\ref{thm:matching} can be used to compare various CSAMs with optimistic and pessimistic agents w.r.t.~worst-case utilitarian rank and worst-case egalitarian rank. 
\begin{dfn}\label{dfn:worstcase} Given any CSAM $f_\mo$ and any $n$, the {\em worst-case utilitarian rank} is $\max_{P_n}\sum_{R_j\in {P_n}}\rank(R_j,f_\mo^j({P_n}))$, and the {\em worst-case egalitarian rank} is $\max_{P_n}\max_{R_j\in {P_n}}\rank(R_j,f_\mo^j({P_n}))$, where ${P_n}$ is a profile of $n$ agents.
\end{dfn}
In words, the worst-case utilitarian rank is the worst (largest) total rank of the bundles (w.r.t.~respective agent's preferences) allocated by $f_\mo$. The worst-case egalitarian rank is the worst (largest) rank of the least-satisfied agent, which is also a well-accepted measure of fairness. The worst case is taken over all profiles of $n$ agents. 

\begin{prop}\label{prop:optsd} Among all CSAMs, serial dictatorships with all-optimistic agents have the best (smallest) worst-case utilitarian rank and the worst (largest) worst-case egalitarian rank.
\end{prop} 
\begin{proof} The worst-case egalitarian rank of any serial dictatorship is $n^p$ when all agents have the same preferences. To prove the optimality of worst-case utilitarian rank, given $f_\mo$, we consider the multiset composed of the numbers of items in the designated category that the active agent can choose in each step. That is, we consider the multiset $\text{RI}=\{k_{j,l}:\forall j\leq n, l\leq p\}$. Since in each step in the execution of $f_\mo$, only one item is allocated, RI is composed of $p$ copies of $\{1,\ldots,n\}$. Since for each agent $j$, $(1+\sum_{l=1}^p(k_{j,O_j(l)}-1))\leq \prod_{l=1}^pk_{j,O_j(l)}$ (we note that in the right hand side, $l$ starts with $1$ but not $K_j$), the best worst-case utilitarian rank is at least $n(n^{p}+1)-\sum_{j=1}^n\prod_{l=1}^pk_{j,O_j(l)}\geq n(n^{p}+1)-\sum_{j=1}^nj^p$. It is not hard to verify that this lower bound is achieved by any serial dictatorship with all-optimistic agents.
\end{proof}

\begin{prop}\label{prop:csmopt} Any CSAM with all-optimistic agents has the worst (largest) worst-case egalitarian rank, which is $n^p$.
\end{prop}
\begin{proof} By Theorem~\ref{thm:matching}, the proposition is equivalent to the existence of an agent $j$ such that for all $l\geq K_{j}$, $k_{j,\mo_j(l)}=1$. For the sake of contradiction, let us assume the following condition:

{\bf Condition (*)}: for every agent $j$, there exists $l\geq K_{j}$ such that $k_{j,\mo_j(l)}>1$.

Let $\mo(np)=(j_n,i_p)$.  It follows that $k_{j_n,i_p}=1$ because there is only one item left. By condition (*), there exists $i_{p-1}$ with $K_{n}\leq \mo_{j_n}^{-1}(i_{p-1})$ such that $k_{j_n,i_{p-1}}>1$. 
Let $j_{n-1}$  denote the agent who is the last to choose an item from category $i_{p-1}$. We have $j_{n-1}\neq j_n$, because agent $n$ is not the last agent to choose an item from category $i_{p-1}$. By definition, we have $k_{j_{n-1},i_{p-1}}=1$. Moreover, $(j_{n-1},\mo_{j_{n-1}}(K_{j_{n-1}}))\rhd_\mo (j_{n},i_{p-1})\rhd_\mo(j_{n},\mo_{j_{n}}(K_{j_{n}}))$, which simply states that $j_{n-1}$ chooses an item from category $\mo_{j_{n-1}}(K_{j_{n-1}})$ after $j_{n}$ chooses an item from category $i_{p-1}$ (the second half of the inequality is due to the way we choose $i_{p-1}$). This inequality holds because if agent $j_{n-1}$ chooses an item from category $\mo_{j_{n-1}}(K_{j_{n-1}})$ before agent  $j_{n}$ chooses an item from category $i_{p-1}$, then agent $j_n$ ``interrupts'' agent $j_{n-1}$ from choosing an item from category $i_{p-1}$, which contradicts the definition of $K_{j_{n-1}}$. 

By condition~(*), there exists $i_{p-2}$ such that $K_{j_{n-1}}\leq \mo_{j_{n-1}}^{-1}(i_{p-2})$ and $k_{j_{n-1},i_{p-2}}>1$. Similarly, we can define $j_{n-2}$, prove that $j_{n-2}\neq j_{n-1}$ and $(j_{n-2},\mo_{j_{n-2}}(K_{j_{n-2}}))\rhd_\mo (j_{n-1},i_{p-2})\rhd_\mo(j_{n-1},\mo_{j_{n-1}}(K_{j_{n-1}}))$. 

However, this process cannot continue forever, since otherwise we will obtain an infinite sequence in $\mo$: $(j_{n},\mo_{j_{n}}(K_{j_{n}}))\lhd_\mo(j_{n-1},\mo_{j_{n-1}}(K_{j_{n-1}}))\lhd_\mo(j_{n-2},\mo_{j_{n-2}}(K_{j_{n-2}}))\lhd_\mo\cdots$, but $np$ is finite. This leads to a contradiction. 
\end{proof}

\begin{prop}\label{prop:bcsm}  For any even number $p$, the worst-case egalitarian rank of any balanced CSAM (defined in Example~\ref{ex:csm}) with all-pessimistic agents  is $n^p-(n-1)p/2$. These are the CSAMs with the best worst-case egalitarian rank among CSAMs with all-pessimistic agents.
\end{prop}
\begin{proof} For any balanced CSAM, it is not hard to see that for any agent $j$, $\mo_j=1\rhd\cdots\rhd p$. For any $l<p/2$ and any $j\leq n$, we have $k_{j,2l-1}+k_{j,2l}=n+1$. Since all agents are pessimistic, by part 2 of Theorem~\ref{thm:matching}, their worst-case ranks are all equal to $n^p-(n-1)p/2$. The optimality of balanced CSAMs comes from the fact that for any categorial sequential mechanisms $\sum_{j,l}k_{j,l}=(n+1)np/2$. Therefore, there must exist an agent $j^*$ with $\sum_{l=1}^pk_{j,l}\leq (n+1)p/2$.\end{proof}

A natural question after Proposition~\ref{prop:bcsm} is: do the balanced CSAMs with all-pessimistic agents have optimal worst-case egalitarian rank, among all CSAMs {\em for any combination of optimistic and pessimistic agents}? The answer is negative. 

\begin{prop}\label{prop:wer} For any even number $p$ with $2^p>1+(n-1)p/2$, there exists a CSAM with both optimistic and pessimistic agents, whose worst-case egalitarian rank is strictly better (smaller) than $n^p-(n-1)p/2$.
\end{prop}
\begin{proof} We prove the proposition by explicitly constructing such a mechanism. The idea is, agents $\{1,\ldots,n-1\}$ choose the items as in a balanced CSAM for $n-1$ agents, then we let agent $n$ ``interrupt'' them and choose all items in consecutive $p$ rounds right before their last iteration, i.e.~the last $(n-1)$ round. Then, we let agents $1$ through $n-1$ be optimistic and let agent $n$ be pessimistic. For example, when $n=3$ and $p=4$, the order is $(1,1)\rhd(2,1)\rhd(2,2)\rhd(1,2)\rhd(1,3)\rhd(2,3)\rhd (3,1)\rhd(3,2)\rhd(3,3)\rhd(3,4)\rhd(2,4)\rhd(1,4)$. Agent $1$ and agent $2$ are optimistic and agent $3$ is pessimistic.

By part 2 of Theorem~\ref{thm:matching}, for any agent $j\leq n-1$, the worst-case rank is $n^p+1-(1+np/2)$. By part 1 of Theorem~\ref{thm:matching}, the worst-case rank for agent $n$ is $n^p+1-2^p$. This proves the proposition.
\end{proof}

\section{Rank Efficiency of  CSAMs for Strategic Agents}
When all agents are strategic and have complete and perfect information, the CSAM naturally corresponds to an extensive-form game, for which we will focus on the subgame-perfect Nash Equilibrium (SPNE), which is unique because agents' preferences are linear orders.
\begin{ex}\label{ex:strategic} \rm Let $n=p=2$. Two agents' preferences are $R_1 = [22\succ 12\succ 11\succ 21]$ and $R_2 = [12\succ 11\succ 22\succ 21]$. Let $\mo = (1,1)\rhd (2,2)\rhd (2,1)\rhd (1,2)$. Suppose both agents are strategic and have complete and perfect information. The game tree and backward induction for SPNE are illustrated in Figure~\ref{fig:tree} (a). When $t=3,4$, the active agent only has one choice. When $t=2$, on the left tree agent $2$ chooses item $2$ from $D_2$ because she prefers $22$ to $21$; on the right tree agent $2$ chooses item $2$ from $D_2$ because she prefers $12$ to $11$. When $t=1$, agent $1$ chooses item $1$ from $D_1$ because she prefers $11$ to $21$.\hfill$\Box$
\end{ex}

\begin{figure}[htp]
\centering
\begin{tabular}{cc}
\includegraphics[trim=0 10cm 12cm 0, clip=true, width=.55\textwidth]{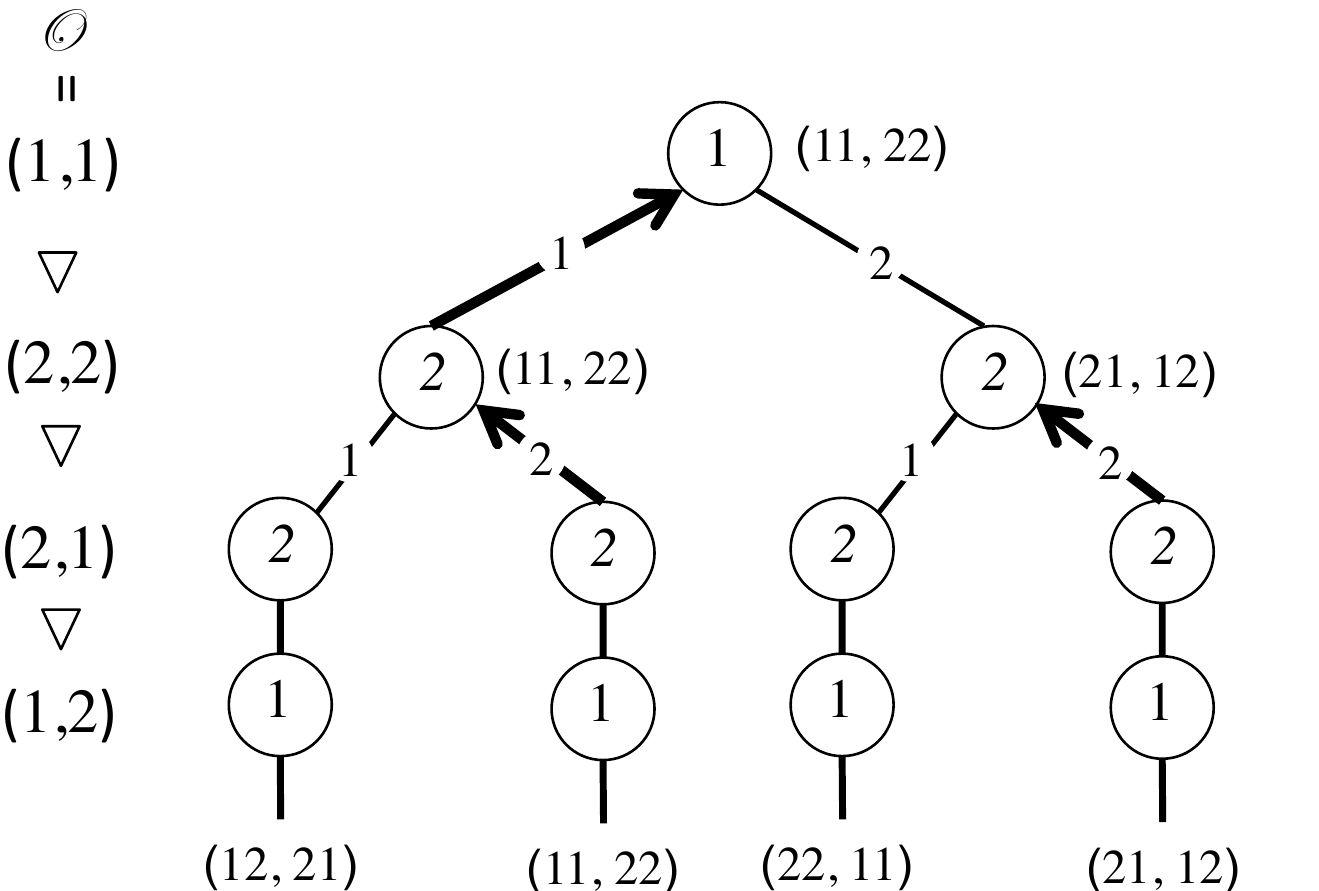}&
\includegraphics[trim=17cm 12cm 0cm 0, clip=true,width=.4\textwidth]{tree.pdf}\\
(a) & (b)
\end{tabular}
\caption{\small (a) Extensitve-form game in Example~\ref{ex:strategic}. (b) Illustration of proof for Proposition~\ref{prop:strategic}.\label{fig:tree}}
\end{figure}
\begin{prop}\label{prop:strategic} Suppose all agents are strategic and have complete and perfect information. 
For any CSAM $f_\mo$, any $j\leq n$, and any profile, the rank of the bundle allocated to agent $j$ in the SPNE is at most $n^p+1-\prod_{i\leq p}k_{j,i}$.
\end{prop}
\begin{proof} For any $j\leq n$ and $t\leq np$, we define $A_{j,t}$ recursively as follows. For all $j$, $A_{j,np} = 1$. For any $2\leq t\leq np$, suppose $\mo(t) = (j,i)$. We let $A_{j,t-1}=k_{j,i}\times A_{j,t}$ and for any $j'\neq j$, we let $A_{j',t-1}=A_{j',t}$.

We next prove by induction that for any node $\mn$ (corresponding to step $t$) in the backward induction tree of the extensive-form game of $f_\mo$, any $j\leq n$, and any $t\leq np$, there are at least $A_{j,t}-1$  bundles at the leaves of the subtree rooted at $\mn$ that are ranked below the bundle agent $j$ chooses at the node w.r.t.~agent $j$'s preferences.

The case for all nodes with $t = np$ naturally holds. Supposing the claim holds for all nodes with  $2\leq t'\leq np$, we next show that it also holds for all nodes with $t'-1$. For any node $\mn$ with $t'-1$, let agent $j$ denote the agent who will chose an item from category $i$. By definition, agent $j$ has $k_{j,i}$ choices, which corresponds to the $k_{j,i}$ subtrees of $\mn$. Let $B= \{b_1\ldots b_{k_{j,i}}\}$ denote these bundles from subtree $1$ through $k_{j,i}$ respectively. See Figure~\ref{fig:tree} (b) for illustration. By the induction hypothesis, for each subtree $l$, there are at least $A_{j,t'} -1$ bundles ranked below $b_l$. Because agent $j$ will choose her most preferred bundle in $B$, w.l.o.g.~let it be $b_1$, there are at least $k_{j,i}(A_{j,t'} -1)+k_{j,i}-1= A_{j,t'-1}-1$ bundles ranked below $b_1$ among all bundles for $j$ in the leaves of $\mn$. For any other agent $j'$, her allocation is the same as in subtree $1$, which means that there are at least $A_{j',t'}-1 = A_{j',t'-1}-1$ bundles in the leaves of subtree $1$ ranked below the bundle allocated to her at $\mn$. This proves that the claim holds for all $t$. 

The proposition follows after the observation that for any $j\leq n$, $A_{j,1} = \prod_{i\leq p}k_{j,i}$.
\end{proof}

The next proposition shows that when there are two agents, for any $p$ and any CSAM, the bounds in Proposition~\ref{prop:strategic} can be reached in a same profile. The proof is by construction and is relegated to the appendix.
\begin{prop}\label{prop:strategicn=2} When $n=2$, both agents are strategic, and have complete and perfect information. For any CSAM $f_\mo$, there exists a profile $P$ such that for $j= 1,2$, the rank of the bundle allocated to agent $j$ is $n^p+1-\prod_{i\leq p}k_{j,i}$.
\end{prop}

\newcommand{\myheight}{.26\textheight}
{


}

\section{Simulation Results}
In this section, we use computer simulations to evaluate {\em expected} efficiency of categorial sequential mechanisms, when agents' preferences are generated i.i.d.~from a well-known statistical model called the {\em Mallows model}~\cite{Mallows57:Non-null}. In a Mallows model, we are given a dispersion parameter $0<\varphi\leq 1$, and for any ground truth ranking $W$ and any ranking $V$, we have $\Pr(V|W)=\frac{1}{Z}\cdot \varphi^{\kendall(V,W)}$, where $\kendall(V,W)$ is the {\em Kendall-tau distance} between $V$ and $W$, defined to be the number of different pairwise comparisons between alternatives and $Z$ is the normalization factor. In the Mallows model, the dispersion parameter measures the centrality of the generated linear orders. The smaller $\varphi$ is, the more centralized the randomly generated linear orders are (around the ground truth linear order). When $\varphi=1$, the Mallows model degenerates to the uniform distribution for any ground truth linear order $W$.


{\noindent\bf Data generation.} We fix $p=2$, let $n$ range from $2$ to $11$, and let $\varphi$ be $0.1$, $0.5$, and $1$. For each setting, we first randomly generate a linear order $W$ over $\md$, and then use it as the ground truth in the Mallows model to generate $n$ agents' preferences. For each setting we generate $2000$ datasets and use them to approximately compute the expected utilitarian rank and the expected egalitarian rank, defined by replacing $\max_{P_n}$ by $E_{P_n}$ in Definition~\ref{dfn:worstcase}.\footnote{The expected egalitarian rank should be distinguished from the {\em egalitarian expected rank}, which first computes the expected rank for every agent, then chooses the largest (expected) rank.} We evaluate serial dictatorships and balanced CSAMs with two configurations of agents: all-optimistic agents and all-pessimistic agents.  All computations were done on a 1.8 GHz Intel Core i7 laptop with 4GB memory.

{\noindent\bf Results.} The results for $\varphi=0.5$ are shown in Figure~\ref{fig:phi=0.5}, where we also plot 95\% confidence intervals. Results for other $\varphi$'s are similar. 
We observe that in general, serial dictatorships with all-optimistic agents have the best (smallest) expected utilitarian rank, and balanced CSAMs with all-pessimistic agents have the best (smallest) expected egalitarian rank. All these comparisons are statistically significant at the 0.05 level.
These observations are incidentally consistent with the worst-case results obtained in Section~\ref{sec:eff}.

\begin{figure*}[htp]\centering
\begin{tabular}{cc}
\includegraphics[trim=6mm 1mm 16mm 10mm, clip=true, width=.49\textwidth,height=\myheight]{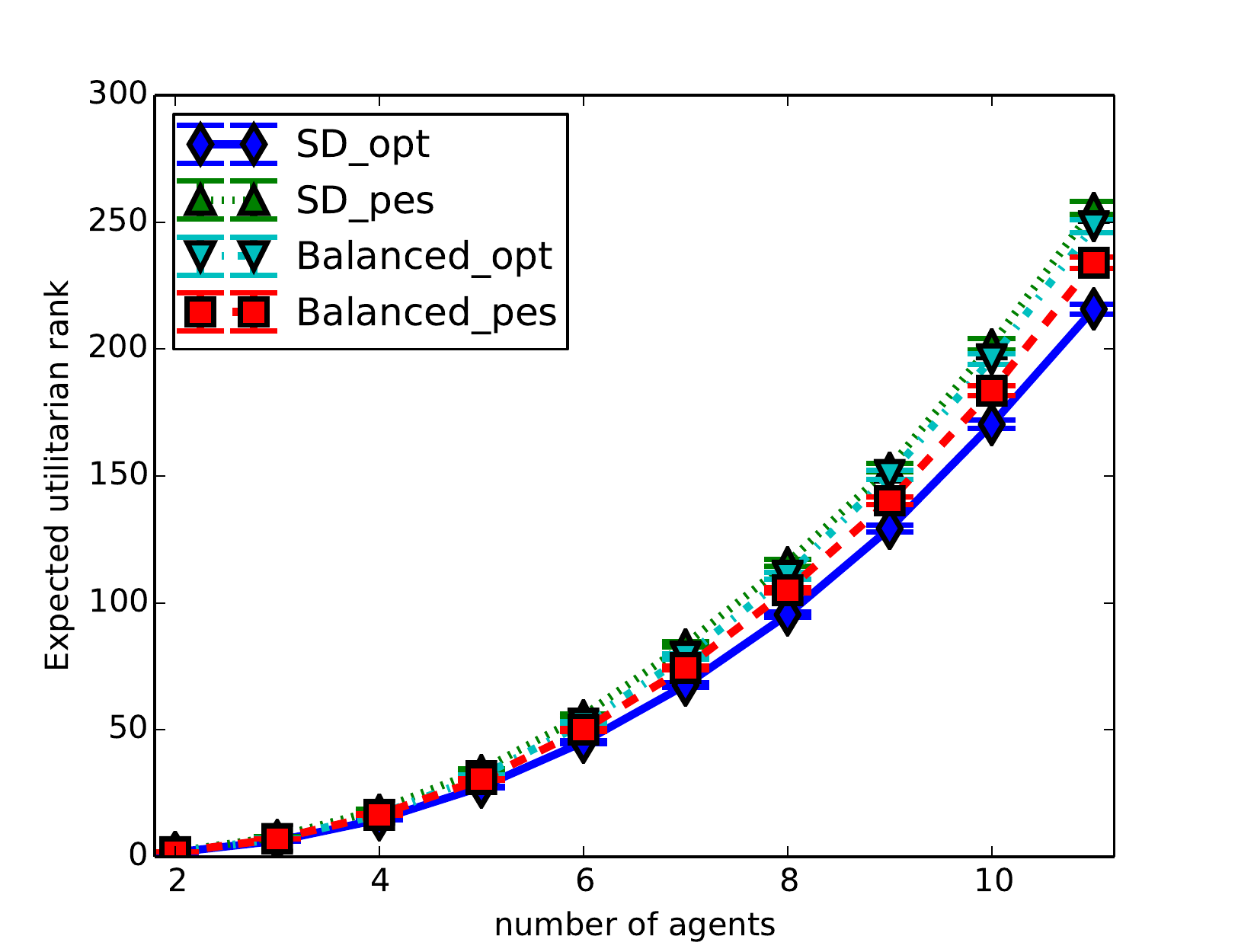}&\includegraphics[trim=7mm 1mm 16mm 10mm,  clip=true,width=.49\textwidth,height=\myheight]{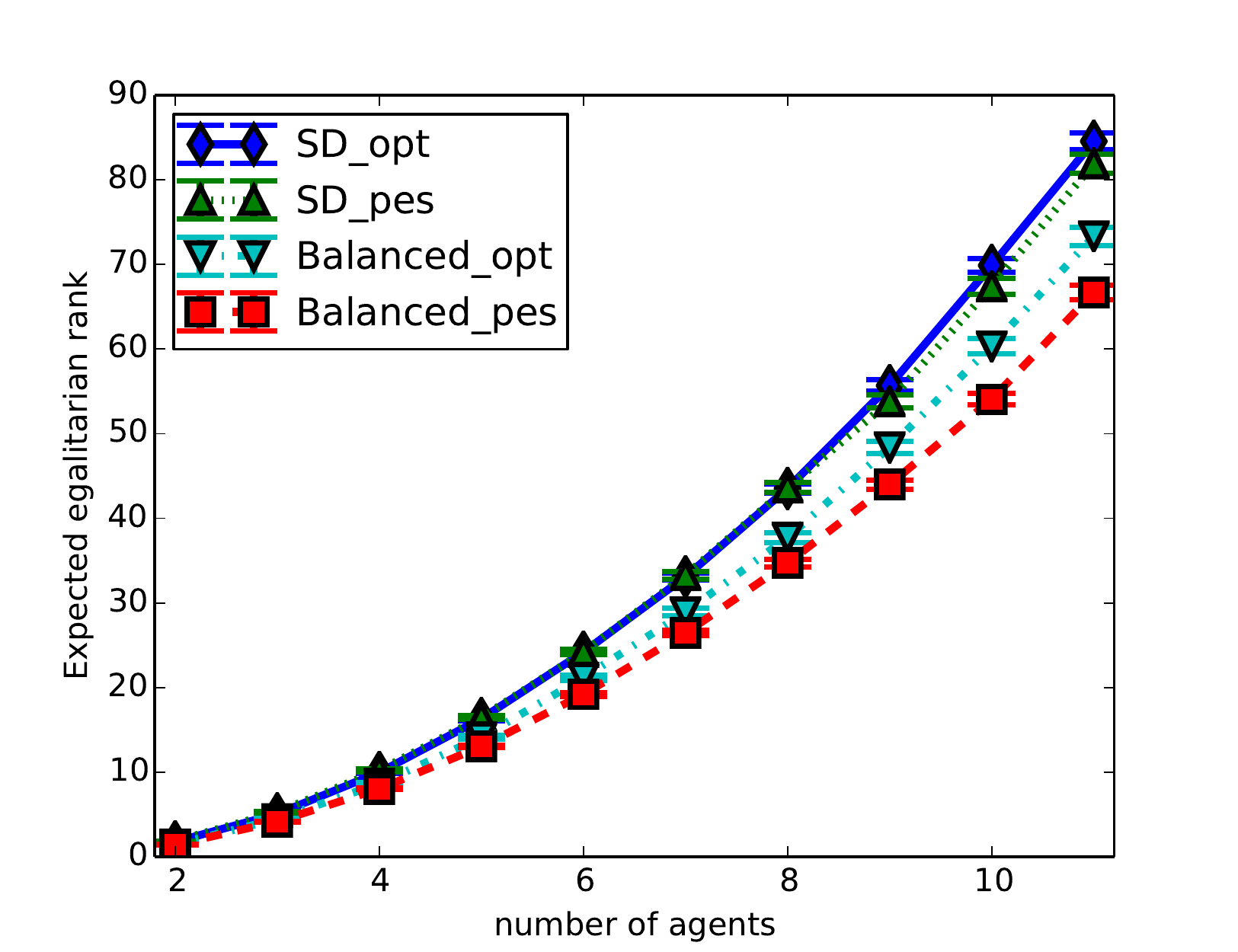}
\\
\small Expected utilitarian rank. &\small  Expected egalitarian rank.
\end{tabular}
\caption{ The data are generated from the Mallows model with $\varphi=0.5$.\label{fig:phi=0.5}}
\end{figure*}

\section{Summary and Future Work}
In this paper we propose CDAPs to model allocation problems for indivisible and categorized items without monetary transfer, when agents have heterogenous and combinatorial preferences. We characterize serial dictatorships for basic CDAPs, propose CSAMs and characterize worst-case rank efficiency for CSAMs with any combination of optimistic and pessimistic agents, which leads to characterizations of utilitarian rank and egalitarian rank of various CSAMs. We provide an upper bound on the rank efficiency for CSAMs with strategic agents and a matching lower bound for the case of two agents. We also performed some preliminary simulations to compare the expected ranks of various CSAMs.

There are many open questions and directions for future research, including proving a matching bound on the rank efficiency for $n>2$ strategic agents, analyzing the outcomes  and rank efficiency for CSAMs for other types of agents, e.g.~minimax-regret agents. We also plan to 
work on theoretical analysis of the expected utilitarian rank and egalitarian rank for randomized allocation mechanisms. For general CDAPs, we are excited to explore generalizations of CP-nets~\cite{Boutilier04:CP}, LP-trees~\cite{Booth10:Learning}, and soft constraints~\cite{Pozza11:Multi-agent} for preference representation. Based on these new languages we can analyze fairness and computational aspects of CSAMs and other mechanisms. Mechanism design for CDAPs with sharable, non-sharable, and divisible items is also an important and promising topic for future research.

%
{\small

\begin{thebibliography}{}

\bibitem[\protect\citeauthoryear{Bhattacharya, Culler, Friedman, Ghodsi,
  Shenker, and Stoica}{Bhattacharya
  et~al\mbox{.}}{2013}]{Bhattacharya13:Hierarchical}
{\sc Bhattacharya, A.~A.}, {\sc Culler, D.}, {\sc Friedman, E.}, {\sc Ghodsi,
  A.}, {\sc Shenker, S.}, {\sc and} {\sc Stoica, I.} 2013.
\newblock Hierarchical scheduling for diverse datacenter workloads.
\newblock In {\em Proceedings of the 4th Annual Symposium on Cloud Computing}.
  Santa Clara, CA, USA, 4:1--4:15.

\bibitem[\protect\citeauthoryear{Bogomolnaia and Moulin}{Bogomolnaia and
  Moulin}{2001}]{Bogomolnaia01:New}
{\sc Bogomolnaia, A.} {\sc and} {\sc Moulin, H.} 2001.
\newblock {A New Solution to the Random Assignment Problem}.
\newblock {\em Journal of Economic Theory\/}~{\em 100,\/}~2, 295--328.

\bibitem[\protect\citeauthoryear{Booth, Chevaleyre, Lang, Mengin, and
  Sombattheera}{Booth et~al\mbox{.}}{2010}]{Booth10:Learning}
{\sc Booth, R.}, {\sc Chevaleyre, Y.}, {\sc Lang, J.}, {\sc Mengin, J.}, {\sc
  and} {\sc Sombattheera, C.} 2010.
\newblock Learning conditionally lexicographic preference relations.
\newblock In {\em Proceeding of the 2010 conference on ECAI 2010: 19th European
  Conference on Artificial Intelligence}. Amsterdam, The Netherlands, 269--274.

\bibitem[\protect\citeauthoryear{Boutilier, Brafman, Domshlak, Hoos, and
  Poole}{Boutilier et~al\mbox{.}}{2004}]{Boutilier04:CP}
{\sc Boutilier, C.}, {\sc Brafman, R.}, {\sc Domshlak, C.}, {\sc Hoos, H.},
  {\sc and} {\sc Poole, D.} 2004.
\newblock {CP}-nets: {A} tool for representing and reasoning with conditional
  ceteris paribus statements.
\newblock {\em Journal of Artificial Intelligence Research\/}~{\em 21},
  135--191.

\bibitem[\protect\citeauthoryear{Boutilier, Caragiannis, Haber, Lu, Procaccia,
  and Sheffet}{Boutilier et~al\mbox{.}}{2012}]{Boutilier12:Optimal}
{\sc Boutilier, C.}, {\sc Caragiannis, I.}, {\sc Haber, S.}, {\sc Lu, T.}, {\sc
  Procaccia, A.~D.}, {\sc and} {\sc Sheffet, O.} 2012.
\newblock {Optimal social choice functions: A utilitarian view}.
\newblock In {\em ACM Conference on Electronic Commerce}. Valencia, Spain,
  197--214.

\bibitem[\protect\citeauthoryear{Bouveret, Endriss, and Lang}{Bouveret
  et~al\mbox{.}}{2009}]{Bouveret09:Conditional}
{\sc Bouveret, S.}, {\sc Endriss, U.}, {\sc and} {\sc Lang, J.} 2009.
\newblock Conditional importance networks: {A} graphical language for
  representing ordinal, monotonic preferences over sets of goods.
\newblock In {\em Proceedings of the 21st international jont conference on
  Artifical intelligence}. IJCAI'09. Pasadena, California, USA, 67--72.

\bibitem[\protect\citeauthoryear{Bouveret and Lang}{Bouveret and
  Lang}{2011}]{Bouveret11:General}
{\sc Bouveret, S.} {\sc and} {\sc Lang, J.} 2011.
\newblock A general elicitation-free protocol for allocating indivisible goods.
\newblock In {\em Proceedings of the Twenty-Second International Joint
  Conference on Artificial Intelligence (IJCAI)}. Barcelona, Catalonia, Spain,
  73--78.

\bibitem[\protect\citeauthoryear{Brams, Jones, and Klamler}{Brams
  et~al\mbox{.}}{2006}]{Brams06:Better}
{\sc Brams, S.~J.}, {\sc Jones, M.~A.}, {\sc and} {\sc Klamler, C.} 2006.
\newblock {Better Ways to Cut a Cake}.
\newblock {\em Notices of the AMS\/}~{\em 53,\/}~11, 1314--1321.

\bibitem[\protect\citeauthoryear{Brams, Kilgour, and Zwicker}{Brams
  et~al\mbox{.}}{1998}]{Brams98:Paradox}
{\sc Brams, S.~J.}, {\sc Kilgour, D.~M.}, {\sc and} {\sc Zwicker, W.~S.} 1998.
\newblock The paradox of multiple elections.
\newblock {\em Social Choice and Welfare\/}~{\em 15,\/}~2, 211--236.

\bibitem[\protect\citeauthoryear{Brandt, Conitzer, and Endriss}{Brandt
  et~al\mbox{.}}{2013}]{Brandt13:Computational}
{\sc Brandt, F.}, {\sc Conitzer, V.}, {\sc and} {\sc Endriss, U.} 2013.
\newblock Computational social choice.
\newblock In {\em Multiagent Systems}, {G.~Weiss}, Ed. MIT Press.

\bibitem[\protect\citeauthoryear{Budish and Cantillon}{Budish and
  Cantillon}{2012}]{Budish12:Multi}
{\sc Budish, E.} {\sc and} {\sc Cantillon, E.} 2012.
\newblock {The Multi-Unit Assignment Problem: {T}heory and Evidence from Course
  Allocation at Harvard}.
\newblock {\em American Economic Review\/}~{\em 102,\/}~5, 2237--71.

\bibitem[\protect\citeauthoryear{Chevaleyre, Dunne, Endriss, Lang, Lemaitre,
  Maudet, Padget, Phelps, Rodr\'{i}guez-Aguilar, and Sousa}{Chevaleyre
  et~al\mbox{.}}{2006}]{Chevaleyre06:Issues}
{\sc Chevaleyre, Y.}, {\sc Dunne, P.~E.}, {\sc Endriss, U.}, {\sc Lang, J.},
  {\sc Lemaitre, M.}, {\sc Maudet, N.}, {\sc Padget, J.}, {\sc Phelps, S.},
  {\sc Rodr\'{i}guez-Aguilar, J.~A.}, {\sc and} {\sc Sousa, P.} 2006.
\newblock Issues in multiagent resource allocation.
\newblock {\em Informatica\/}~{\em 30}, 3--31.

\bibitem[\protect\citeauthoryear{Ehlers and Klaus}{Ehlers and
  Klaus}{2003}]{Ehlers03:Coalitional}
{\sc Ehlers, L.} {\sc and} {\sc Klaus, B.} 2003.
\newblock Coalitional strategy-proof and resource-monotonic solutions for
  multiple assignment problems.
\newblock {\em Social Choice Welfare\/}~{\em 21}, 265---280.

\bibitem[\protect\citeauthoryear{Ghodsi, Sekar, Zaharia, and Stoica}{Ghodsi
  et~al\mbox{.}}{2012}]{Ghodsi12:Multi}
{\sc Ghodsi, A.}, {\sc Sekar, V.}, {\sc Zaharia, M.}, {\sc and} {\sc Stoica,
  I.} 2012.
\newblock {Multi-resource Fair Queueing for Packet Processing}.
\newblock In {\em Proceedings of the ACM SIGCOMM 2012 conference on
  Applications, technologies, architectures, and protocols for computer
  communication}. Vol.~42. Helsinki, Finland, 1--12.

\bibitem[\protect\citeauthoryear{Ghodsi, Zaharia, Hindman, Konwinski, Shenker,
  and Stoica}{Ghodsi et~al\mbox{.}}{2011}]{Ghodsi11:Dominant}
{\sc Ghodsi, A.}, {\sc Zaharia, M.}, {\sc Hindman, B.}, {\sc Konwinski, A.},
  {\sc Shenker, S.}, {\sc and} {\sc Stoica, I.} 2011.
\newblock {Dominant Resource Fairness: Fair Allocation of Multiple Resource
  Types}.
\newblock In {\em {Proceedings of the 8th USENIX Conference on Networked
  Systems Design and Implementation}}. {Boston, MA, USA}, 323--336.

\bibitem[\protect\citeauthoryear{Hatfield}{Hatfield}{2009}]{Hatfield09:Strategy-proof}
{\sc Hatfield, J.~W.} 2009.
\newblock Strategy-proof, efficient, and nonbossy quota allocations.
\newblock {\em Social Choice and Welfare\/}~{\em 33,\/}~3, 505--515.

\bibitem[\protect\citeauthoryear{Huh, Liu, and Truong}{Huh
  et~al\mbox{.}}{2013}]{Huh13:Multiresource}
{\sc Huh, W.~T.}, {\sc Liu, N.}, {\sc and} {\sc Truong, V.-A.} 2013.
\newblock {Multiresource Allocation Scheduling in Dynamic Environments}.
\newblock {\em Manufacturing and Service Operations Management\/}~{\em
  15,\/}~2, 280--291.

\bibitem[\protect\citeauthoryear{Koutsoupias and Papadimitriou}{Koutsoupias and
  Papadimitriou}{1999}]{Koutsoupias99:Worst-case}
{\sc Koutsoupias, E.} {\sc and} {\sc Papadimitriou, C.} 1999.
\newblock {Worst-case Equilibria}.
\newblock In {\em Proceedings of the 16th Annual Conference on Theoretical
  Aspects of Computer Science}. Trier, Germany, 404--413.

\bibitem[\protect\citeauthoryear{Lacy and Niou}{Lacy and
  Niou}{2000}]{Lacy00:Problem}
{\sc Lacy, D.} {\sc and} {\sc Niou, E.~M.} 2000.
\newblock A problem with referendums.
\newblock {\em Journal of Theoretical Politics\/}~{\em 12,\/}~1, 5--31.

\bibitem[\protect\citeauthoryear{Mallows}{Mallows}{1957}]{Mallows57:Non-null}
{\sc Mallows, C.~L.} 1957.
\newblock Non-null ranking model.
\newblock {\em Biometrika\/}~{\em 44,\/}~1/2, 114--130.

\bibitem[\protect\citeauthoryear{Meir, Lev, and Rosenschein}{Meir
  et~al\mbox{.}}{2014}]{Meir14:Local}
{\sc Meir, R.}, {\sc Lev, O.}, {\sc and} {\sc Rosenschein, J.~S.} 2014.
\newblock {A Local-Dominance Theory of Voting Equilibria}.
\newblock In {\em Proceedings of the 15th ACM Conference on Electronic
  Commerce}. Palo Alto, CA, USA, 313--330.

\bibitem[\protect\citeauthoryear{P\'apai}{P\'apai}{2000a}]{Papai00:Strategyproof}
{\sc P\'apai, S.} 2000a.
\newblock Strategyproof assignment by hierarchical exchange.
\newblock {\em Econometrica\/}~{\em 68,\/}~6, 1403--1433.

\bibitem[\protect\citeauthoryear{P\'apai}{P\'apai}{2000b}]{Papai00:Strategyproofquotas}
{\sc P\'apai, S.} 2000b.
\newblock Strategyproof multiple assignment using quotas.
\newblock {\em Review of Economic Design\/}~{\em 5}, 91--105.

\bibitem[\protect\citeauthoryear{P\'apai}{P\'apai}{2001}]{Papai01:Strategyproof}
{\sc P\'apai, S.} 2001.
\newblock Strategyproof and nonbossy multiple assignments.
\newblock {\em Journal of Public Economic Theory\/}~{\em 3,\/}~3, 257--71.

\bibitem[\protect\citeauthoryear{Pozza, Pini, Rossi, and Venable}{Pozza
  et~al\mbox{.}}{2011}]{Pozza11:Multi-agent}
{\sc Pozza, G.~D.}, {\sc Pini, M.~S.}, {\sc Rossi, F.}, {\sc and} {\sc Venable,
  K.~B.} 2011.
\newblock Multi-agent soft constraint aggregation via sequential voting.
\newblock In {\em Proceedings of the Twenty-Second International Joint
  Conference on Artificial Intelligence}. Barcelona, Catalonia, Spain,
  172--177.

\bibitem[\protect\citeauthoryear{Procaccia and Rosenschein}{Procaccia and
  Rosenschein}{2006}]{Procaccia06:Distortion}
{\sc Procaccia, A.~D.} {\sc and} {\sc Rosenschein, J.~S.} 2006.
\newblock {The Distortion of Cardinal Preferences in Voting}.
\newblock In {\em Proceedings of the 10th International Workshop on Cooperative
  Information Agents}. LNAI Series, vol. 4149. 317--331.

\bibitem[\protect\citeauthoryear{S\"{o}nmez and \"{U}nver}{S\"{o}nmez and
  \"{U}nver}{2011}]{Sonmez11:Matching}
{\sc S\"{o}nmez, T.} {\sc and} {\sc \"{U}nver, M.~U.} 2011.
\newblock {Matching, Allocation, and Exchange of Discrete Resources}.
\newblock In {\em Handbook of Social Economics}, {J.~Benhabib}, {A.~Bisin},
  {and} {M.~O. Jackson}, Eds. North-Holland, Chapter~17, 781--852.

\bibitem[\protect\citeauthoryear{Svensson}{Svensson}{1999}]{Svensson99:Strategy-proof}
{\sc Svensson, L.-G.} 1999.
\newblock Strategy-proof allocation of indivisible goods.
\newblock {\em Social Choice and Welfare\/}~{\em 16,\/}~4, 557--567.

\end{thebibliography}

}
\newpage
\section*{Appendix: Full Proofs}
\subsection{Proof of Theorem~\ref{thm:matching}}

\begin{proof} Given $\mo$ and the information on whether each agent $j$ is optimistic or pessimistic, we will construct a profile $P$ such that in $f_\mo(P)$, for all $j\leq n$, agent $j$ obtains $(j,\ldots,j)$. 

We prove the theorem in the following three steps: 
\begin{itemize}
\item {\bf Step 1: define bottom bundles}. We specify a set of bundles that are ranked in the bottom positions for each agent $j$, and require $(j,\ldots,j)$ to be ranked on the top of them. 
\item {\bf Step 2: define top bundles}. We specify top-$1$ and sometimes also top-$2$ bundles for each agent. 
\item {\bf Step 3: extend to full profile}. We take a profile that extends the partial orders constructed in the first two steps, and then show that it satisfies all three properties in the theorem. 
\end{itemize}The construction is summarized in Table~\ref{tab:optapp} (for optimistic agents) and Table~\ref{tab:pesapp} (for pessimistic agents).
   \renewcommand{\arraystretch}{1.5}
\begin{table*}[htp]
\tbl{Partial preferences for an optimistic agent $j$.  BottomBundles$_j^\text{Opt}$ is defined in (\ref{equ:bundleopt}). ``Others in BottomBundles$_j^\text{Opt}$''  refers to $[\text{BottomBundles}_j^\text{Opt}\setminus\{(j,\ldots,j)\}]$.\label{tab:optapp}}{
\begin{tabular}{|c|r|c|}
\hline \multicolumn{2}{|c|}{Optimistic agent}& Order \\
\hline \multirow{2}{*}{$j\neq j_1$} & case 1:  $K_j=1$& $ ([Pred_{i_1}(j)]_{i_1},[j]_{-i_1})\succ\cdots\succ (j,\ldots, j)\succ \text{others in BottomBundles}_j^\text{Opt}$\\
\cline{2-3} &case 2: $K_j>1$ & $\begin{array}{r}([Pred_{i_1}(j)]_{i_1},[j]_{-i_1})\succ([Pred_{\mo_j(K_j)}(j)]_{\mo_j(K_j)},[j]_{-\mo_j(K_j)})\\
\succ\cdots\succ (j,\ldots, j)\succ \text{others in BottomBundles}_j^\text{Opt}\end{array}$\\
\hline \multirow{2}{*}{$j=j_1$} & case 1: $K_j=1$& $(j_1,\ldots, j_1)\succ ([L_{i_1}(n)]_{i_1},[j_1]_{-i_1})\succ \text{others}$\\
\cline{2-3}  & case 2: $K_j>1$& \begin{tabular}{r}$([Pred_{\mo_{j}(K_j)}(j_1)]_{\mo_j(K_j)},[j_1]_{-\mo_j(K_j)})\succ ([L_{i_1}(n)]_{i_1},[j_1]_{-i_1})$\\ $\succ\cdots\succ (j_1,\ldots, j_1)\succ \text{others in BottomBundles}_{j_1}^\text{Opt}$\end{tabular}\\
\hline 
\end{tabular}
}
\end{table*}

\begin{table*}[htp]
\tbl{Partial preferences for a pessimistic agent $j$. BottomBundles$_j^\text{Pes}$ is defined in (\ref{equ:bundlepes}).  For  $j\neq j_1$,  ``others in BottomBundles$_j^\text{Pes}$''  refers to $(\text{BottomBundles}_{j_1}^\text{Pes}\setminus\{(j,\ldots,j)\})$. For $j=j_1$, ``others in BottomBundles$_{j_1}^\text{Pes}$''  refers to $(\text{BottomBundles}_{j_1}^\text{Pes}\setminus\{(j_1,\ldots,j_1),([L_{i_1}(n)]_{i_1},[j_1]_{-i_1})\})$.\label{tab:pesapp}}{
\begin{tabular}{|r|c|}
\hline Pessimistic agent& Order\\
\hline  $j\neq j_1$ & $([Pred_{i_1}(j)]_{i_1},[j]_{-i_1})\succ\cdots\succ (j,\ldots, j)\succ \text{others in BottomBundles}_j^\text{Pes}$\\
\hline  $j= j_1$ & $\begin{array}{r}([L_{i_1}(n)]_{i_1},[j_1]_{-i_1})\succ\cdots\succ (j_1,\ldots, j_1)\succ \text{others in BottomBundles}_{j_1}^\text{Pes}\\\succ (L_{i_1}(n),\ldots,L_{i_1}(n))\end{array}$\\
\hline 
\end{tabular}
}
\end{table*}

We first introduce some notation that will be useful to define the profile in Step 1 and Step 2. Let $\mo(1)=(j_1,i_1)$. That is, agent $j_1$ is the first to choose an item in the sequential allocation, and she chooses from category $D_{i_1}$. Let $L_{i_1}$ denote the order over $\{1,\ldots,n\}$ representing the order for the {\em agents} to choose items from $D_{i_1}$ in $\mo$. That is, $j\rhd_{L_{i_1}} j'$ if and only if $(j,i_1)\rhd_\mo (j',i_1)$. By definition we have $j_1=L_{i_1}(1)$. For any $j\leq n$, we let $Pred_{i_1}(j)=L_{i_1}(L_{i_1}^{-1}(j)-1)$ denote the predecessor of agent $j$ in $L_{i_1}$, that is, the latest agent who chose an item from category $i_1$ before agent $j$ chooses from category $i_1$. If $j=1$, then we let the last agent in $L_{i_1}$ be her predecessor, that is, $Pred_{i_1}(1)=L_{i_1}(n)$.

\noindent{{\bf Step 1: define bottom bundles}}. In order to match the upper bounds shown in the proof of Proposition~\ref{prop:worstcase}, the bundles described in the proof of Proposition~\ref{prop:worstcase} must be the {\em only} bundles that are ranked below $(j,\ldots,j)$ by agent $j$. This is the part of the profile we will construct in the first step.

For all $i$ and $t$, we first define $D_{i,t}^*$ to be the subset of $D_i=\{1,\ldots, n\}$ such that $q\in D_{i,t}^*$ if and only if agent $q$ has not chosen an item from $D_i$ before the $t$-th round. By definition, if $\mo(t)=(j,i)$ then $j\in D_{i,t}^*$. Formally,
$$D_{i,t}^*=
\{q\leq n: \mo^{-1}(q,i)\geq t\}$$
We note that $D_{i,t}^*$ is defined solely by $i,t$, and $\mo$, which means that it does not depend on agents' preferences and behavior in previous rounds. Later in this proof we will show that for our constructed profile, in each round $(j,i)$ the active agent $j$ will choose $j$ from $D_i$, so that $D_{i,t}^*$ is the remaining items for category $i$ at the beginning of round $t$ of the sequential allocation. 


For any $1\leq l\leq p$, we let $t_{j,l}^*=\mo^{-1}(j,\mo_j(l))$. That is, $t_{j,l}^*$ is the round where agent $j$ chooses an item from the $l$-th category in $\mo_j$, which is not necessarily category $l$. For each agent $j$ we specify their bottom bundles as follows.
\begin{itemize}
\item If agent $j$ is optimistic, then we let the following bundles be ranked in the bottom of her preferences: 
\begin{equation}\label{equ:bundleopt}
\begin{split}
&\text{BottomBundles}_j^\text{Opt} =\\& (j_{\mo_j(1)},\ldots,j_{\mo_j(K_j-1)})\times \prod_{l=K_j}^pD^*_{{\mo_j(l)},t_{j,l}^*},\end{split}\end{equation}
where $(j,\ldots, j)$ is ranked on the top of these bundles, and the order over the remaining bundles is defined arbitrarily.  It follows that $(j,\ldots, j)$ is ranked in the $(\prod_{l=K_j}^pk_{j,O_j(l)})$-th position from the bottom by agent $j$.

\item If agent $j$ is pessimistic, then we first define the following bundles:
\begin{equation}\label{equ:bundlepes}
\begin{split}
&\text{BottomBundles}_j^\text{Pes} =\\
&\bigcup_{l=1}^p\bigcup\nolimits_{d\in D^*_{{\mo_j(l)},t_{j,l}^*}}\{([d]_{\mo_j(l)},[j]_{-\mo_j(l)})\},\end{split}\end{equation}
where $[j]_{-\mo_j(l)}$ means that all components except the $\mo_j(l)$-th component is $j$. 
Bundles in $\text{BottomBundles}_j^\text{Pes}$ are (partially) ranked as follows: first, $(j,\ldots,j)$ is ranked on the top; then, for any $1\leq l_1<l_2\leq p$ and any $d_1\in D^*_{{\mo_j(l_1)},t_{j,l_1}^*}$ and  $d_2\in D^*_{{\mo_j(l_2)},t_{j,l_2}^*}$ with $d_1\neq j$ and $d_2\neq j$, we rank $([d_1]_{\mo_j(l_1)},[j]_{-\mo_j(l_1)})$ below $([d_2]_{\mo_j(l_2)},[j]_{-\mo_j(l_2)})$. 

\begin{itemize}
\item If $j\neq j_1$, then we simply let $\text{BottomBundles}_j^\text{Pes}$ (with the partial orders specified above) be the bundles ranked in the bottom position.

\item If $j=j_1$, then we move $([Pred_{i_1}(j)]_{i_1},[j]_{-i_1})=([L_{i_1}(n)]_{i_1}, [j_1]_{-i_1})$ to the bottom place  and replace it by $(Pred_{i_1}(j),\ldots,Pred_{i_1}(j))=(L_{i_1}(n),\ldots,L_{i_1}(n))$, and then let these be ranked in the bottom positions of agent $j$'s preferences. That is, the bottom bundles are: $
(j_1,\ldots, j_1)\succ (\text{BottomBundles}_j^\text{Pes}\setminus\{(j_1,\ldots, j_1),([L_{i_1}(n)]_{i_1},[j_1]_{-i_1})\})
\succ (L_{i_1}(n),\ldots,L_{i_1}(n))
$
\end{itemize}

In both cases $(j,\ldots, j)$ is ranked at the $(1+\prod_{l=K_j}^p(k_{j,O_j(l)}-1))$-th position from the bottom.\end{itemize}

\noindent{{\bf Step 2: define top bundles}}.  We now specify the top two bundles (sometimes only the top bundle) for optimistic agents, and show that they are compatible with our constructions in Step 1.  For any optimistic agent $j$:
\begin{itemize}
\item When $j\neq j_1$, there are following two cases:

\begin{itemize}
\item case 1: $K_j=1$. We let $([Pred_{i_1}(j)]_{i_1},[j]_{-i_1})$ be the top-ranked bundle of agent $j$. 
\item case 2: $K_j>1$. We let $([Pred_{i_1}(j)]_{i_1},[j]_{-i_1})$ be the top-ranked bundle of agent $j$. Moreover, if $i_1 \neq \mo_j(K_j)$, then we rank $([Pred_{\mo_j(K_j)}(j)]_{\mo_j(K_j)},[j]_{-\mo_j(K_j)})$ at the second position. We recall that $Pred_{\mo_j(K_j)}(j)$ is the predecessor of $j$ in $L_{\mo_j(K_j)}$, the order for the {\em agents} to choose items from $D_{\mo_j(K_j)}$.
\end{itemize} These do not conflict with the preferences specified in Step 1 because item $Pred_{i_1}(j)$ in $D_{i_1}$ is not available for agent $j$ when she is about to choose an item in $D_{i_1}$, and item $Pred_{\mo_j(K_j)}(j)$ in $D_{\mo_j(K_j)}$ is not available for agent $j$ when she is about to choose an item in $D_{\mo_j(K_j)}$. Hence, none of these bundles are in BottomBundles$_j^\text{Opt}$.
\item When $j=j_1$,  there are the following two cases:
\begin{itemize}
\item case 1: $K_j=1$. Since $(j_1,i_1)=\mo(1)$, for all $i$, $D^*_{i,\mo^{-1}(j,i)}=D_i$, which means that agent $j$ is guaranteed to get her top-ranked bundle after the sequential allocation. In this case we let $(j,\ldots, j)$ be agent $j$'s top-ranked bundle and let $([L_{i_1}(n)]_{i_1},[j]_{-i_1})$ be ranked in agent $j$'s second position.  These do not conflict with the preferences specified in Step 1 because in this case Step 1 only specifies that $(j,\ldots,j)$ be ranked in the top position.
\item case 2:  $K_j>1$. We rank $([Pred_{\mo_j(K_j)}(j)]_{\mo_j(K_j)},[j]_{-\mo_j(K_j)})$ at the top position. We then rank $([L_{i_1}(n)]_{i_1},[j]_{-i_1})$ at the second position. Since $i_1=\mo_j(1)$, we have $\mo_j(K_j)\neq i_1$, otherwise $K_j=1$. Hence, $([Pred_{\mo_j(K_j)}(j)]_{\mo_j(K_j)},[j]_{-\mo_j(K_j)})\neq ([L_{i_1}(n)]_{i_1},[j]_{-i_1})$. These do not conflict with the preferences specified in Step 1 because category $i_1$ is agent $j_1$'s first category in $\mo_{j_1}$, which means that $i_1<K_j$, thus $([L_{i_1}(n)]_{i_1},[j]_{-i_1})\not\in\text{BottomBundles}_{j_1}^\text{Opt}$; also $Pred_{\mo_j(K_j)}(j)$ is not available when agent $j_1$ is about to choose an item for category $\mo_j(K_j)$, which means that $([Pred_{\mo_j(K_j)}(j)]_{\mo_j(K_j)},[j]_{-\mo_j(K_j)})\not\in\text{BottomBundles}_{j_1}^\text{Opt}$.
\end{itemize}
\end{itemize}

For any pessimistic agent $j$, we simply let her top-ranked bundle be $([Pred_{i_1}(j)]_{i_1},[j]_{-i_1})$ (we recall that $Pred_{i_1}(j_1)=L_{i_1}(n)$). We claim that preferences specified in the second step do not conflict with preferences specified in the first step for bottom bundles.  
\begin{itemize}
\item If $j\neq j_1$, then we need to show that $([Pred_{i_1}(j)]_{i_1},[j]_{-i_1})\not\in\text{BottomBundles}_j^\text{Pes}$. When agent $j$ is about to choose her item from $D_{i_1}$, agent $Pred_{i_1}(j)$ has already chosen her item from $D_{i_1}$, which means that $Pred_{i_1}(j)$ is unavailable for agent $j$. This means that $([Pred_{i_1}(j)]_{i_1},j_{-i_1})\not\in\text{BottomBundles}_j^\text{Pes}$.
\item If $j=j_1$, then by definition (see Table~\ref{tab:pesapp}) $([L_{i_1}(n)]_{i_1},[j_1]_{-i_1})$ is replaced by $(L_{i_1}(n),\ldots,L_{i_1}(n))$ in {BottomBundles}$_{j_1}^\text{Pes}$, which means that it can be ranked in the top.
\end{itemize}

\noindent{{\bf Step 3: extend to full profile}}. For any $j$, let $R_j$ be an arbitrary linear order over  $\md$ that satisfies all constraints defined in the previous two steps (see Table~\ref{tab:optapp} and~\ref{tab:pesapp}). Let $P=(R_1,\ldots,R_n)$. 

We now show by induction on the round in the sequential allocation mechanism, denoted by $t$, that if we apply the sequential allocation $\mo$ to $P$, then for all $j\leq n$, agent $j$ gets $(j,\ldots, j)$. 

When $t=1$, agent $j_1$ chooses an item from $D_{i_1}$. If $j_1$ is optimistic, then it is not hard to check that the $i_1$-th component of the top-ranked bundle of $R_{j_1}$ is $j_1$ (the top-ranked bundles are $(j,\ldots, j)$ and $([j']_{\mo_j(K_j)},[j]_{-\mo_j(K_j)})$, for case 1 ($K_{j_1}=1$) and case 2 ($K_{j_1}>1$), respectively. If agent $j_1$ is pessimistic, then for any $d\in D_{i_1}$ with $d\neq j_1$, there exists a bundle whose $i_1$th component is $d$ and is ranked below any bundle whose $i_1$th component is $j_1$. More precisely, if $d\neq Pred_{i_1}(j_1)=L_{i_1}(n)$, then such a bundle is $([d]_{i_1},[j]_{-i_1})$; if $d=Pred_{i_1}(j_1)=L_{i_1}(n)$, then such a bundle is $(L_{i_1}(n),\ldots,L_{i_1}(n))$. In both cases a pessimistic agent $j_1$ will choose item $j_1$ from $D_{i_1}$.

Suppose in every round before round $t$, the active agent $j$ chose item $j$ from the designated category. Let $\mo(t)=(j,i)$. If $j$ is optimistic, then we show in the following four cases that she will choose item $j$ from $D_i$ in round $t$. 
\begin{itemize}
\item $j\neq j_1$, $K_j=1$. In this case $j$ is guaranteed to get her top-ranked available bundle. It is not hard to check that the available bundles are a subset of BottomBundles$_j^\text{Opt}$, where $(j,\ldots,j)$ is available and is ranked in the top. Therefore agent $j$ will choose item $j$.
\item $j\neq j_1$, $K_j>1$. There are the following cases: 
\begin{enumerate}
\item Agent $Pred_{i_1}(j)$ has not chosen her item from $D_{i_1}$. In this case the top-ranked bundle $([Pred_{i_1}(j)]_{i_1},[j]_{-i_1})$ is still available by the induction hypothesis. 

\item Agent $Pred_{i_1}(j)$ has chosen an item from $D_{i_1}$ and $Pred_{\mo_j(K_j)}(j)$ has not chosen her item from $D_{\mo_j(K_j)}$. By the induction hypothesis, agent $Pred_{i_1}(j)$ chose item $Pred_{i_1}(j)$  from category $D_{i_1}$, which means that $([Pred_{i_1}(j)]_{i_1},[j]_{-i_1})$ is unavailable. The bundle $([Pred_{\mo_j(K_j)}(j)]_{\mo_j(K_j)},[j]_{-\mo_j(K_j)})$ becomes the top-ranked available bundle due to the induction hypothesis, whose $j$-th component is $j$.

\item $Pred_{i_1}(j)$ has chosen item $Pred_{i_1}(j)$ from $D_{i_1}$ and $Pred_{\mo_j(K_j)}(j)$ has chosen her item from $D_{\mo_j(K_j)}$. In this case, we first claim that $\mo_j^{-1}(i)\geq K_j$. For the sake of contradiction suppose $\mo_j^{-1}(i)< K_j$. Then, by the definition of $Pred_{\mo_j(K_j)}$, no agent chooses an item from $D_{\mo_j(K_j)}$ between round $\mo_j^{-1}(i)$ and $t_{j,K_j}^*$. We recall that $t_{j,K_j}^*$ is the round when agent $j$ chooses an item from $D_{\mo_j(K_j)}$. However, this violates the minimality of $K_j$ since no agent chooses an item from $D_{\mo_j(K_j)}$ between round $t_{j,K_j-1}^*>\mo_j^{-1}(i)$ and $t_{j,K_j}^*$. Hence, we must have that $\mo_j^{-1}(i)\geq K_j$. By the induction hypothesis, the available bundles are a subset of BottomBundles$_j^\text{Opt}$ and $(j,\ldots,j)$ is still available and is ranked at the top, which means that agent $j$ will choose item $j$ from $D_i$.
\end{enumerate}

In all  three cases above, the $i$th component of the top-ranked available bundle is $j$, which means that agent $j$ will choose item $j$.

\item $j= j_1$, $K_j=1$. By the induction hypothesis, the top-ranked bundle $(j,\ldots,j)$ is still available, which means that agent $j$ will choose item $j$.
\item $j= j_1$, $K_j>1$. If agent $Pred_{\mo_j(K_j)}(j)$ has not chosen her item from $D_{\mo_j(K_j)}$, then by the induction hypothesis the top bundle $([Pred_{\mo_j(K_j)}(j)]_{\mo_j(K_j)},[j]_{-\mo_j(K_j)})$ is still available and $i\neq \mo_j(K_j)$.  If agent $Pred_{\mo_j(K_j)}(j)$ has chosen item $Pred_{\mo_j(K_j)}(j)$ from $D_{\mo_j(K_j)}$, then by the induction hypothesis the available bundles are a subset of BottomBundles$_j^\text{Opt}$ with $(j,\ldots,j)$ ranked at the top. In both cases the $i$th component of the top-ranked available bundle is $j$. Therefore agent $j$ will choose item $j$.
\end{itemize}

If agent $j$ is pessimistic, then by the induction hypothesis the available items in $D_i$ are $D_{i,t}^*$, and $j\in D_{i,t}^*$. For any $d\in D_{i,t}^*$ with $d\neq j$, $([d]_i,[j]_{-i})$ is still available and is ranked lower than any {\em available} bundle whose $i$-th component is $j$ in BottomBundles$_j^\text{Pes}$ . Therefore, a pessimistic agent $j$ will choose item $j$ in this round.


It follows that after the sequential allocation, for all $j\leq n$,  agent $j$ gets $(j,\ldots,j)$. It is not hard to verify that conditions 1 and 2 hold.

To show that condition 3 holds, consider the allocation where agent $j$ gets $([Pred_{i_1}(j)]_{i_1},[j]_{-i_1})$. In this allocation, all agents except $j_1$ get their top-ranked bundle, and $j_1$ gets her top-ranked bundle (if $j_1$ is pessimistic) or second-ranked bundle (if $j_1$ is optimistic). This proves the theorem.
\end{proof}

\subsection{Proof of Proposition~\ref{prop:strategicn=2}}

\begin{proof} We construct $P$ by induction on $p$. When $p=1$ the proposition obviously holds. Suppose the proposition holds for $p'$. For any CSAM $f_\mo$ for $p'+1$ categories, w.l.o.g.~in the first round agent $1$ chooses from category $p'+1$. Let $\mo'$ denote the order over categories $1,\ldots, p'$ in $\mo$. $\mo'$ is well-defined because once agent $1$ chooses an item from $D_{p'+1}$, agent $2$'s item from $D_{p'+1}$ is determined, so that the position of $(2,p'+1)$ in $\mo$ does not matter. By the induction hypothesis, let $R_1' = [T_1'\succ b_1'\succ B_1']$ and $R_2'=[T_2'\succ b_2'\succ B_2']$ denote profiles over $D_1\times\cdots\times D_{p'}$ such that $b_1'$ (respectively, $b_2'$) is the bundle allocated to agent $1$ (respectively, agent $2$) in the SPNE of  $f_{\mo'}$, and $|B_1'|+1=\sum_{i=1}^{p'}k_{1,i}$, $|B_2'|+1=\sum_{i=1}^{p'}k_{2,i}$.

Let $R_1 = [(T_1',1)\succ (T_1',2 )\succ (b_1',1)\succ(b_1',2)\succ(B_1',1)\succ (B_1',2 )]$ and $R_2 = [(T_2',1)\succ (b_2',1)\succ (B_2',1 )\succ (T_2',2 )\succ (b_2',2)\succ (B_2',2 ) ]$, where $(T_1',1)$ is a ranking over $\{(\vec d,1): \vec d\in T_1'\}$ such that $(\vec d,1)\succ (\vec e,1)$ if and only if $\vec d\succ_{T_1'} \vec e$. Therefore, if agent $1$ chooses item $1$ from category $p'+1$ in the first round, then in the remaining rounds both agents act as if their preferences are $R_1'$ and $R_2'$, respectively, and the CSAM is $\mo'$. In this case agent $1$ will get $(b_1',1)$. Similarly, if agent $1$ chooses item $2$ from category $p'+1$ in the first round, then she will get $(b_1',2)$. Since agent $1$ prefers $(b_1',1)$ to $(b_1',2)$, she will choose item $1$ in the first round, and the final allocation is: agent $1$ gets $(b_1',1)$, whose rank is $n^{p'+1}-2|B_1'|-1$, and agent $2$ gets $(b_2',2)$, whose rank is $n^{p'+1}-|B_2'|$. We note that $k_{1,p'+1}=2$ and $k_{2,p'+1}=1$. This proves the proposition for $p'+1$. Therefore, the proposition holds for all $p$ and all CSAMs.
\end{proof}

\end{document}